\documentclass[12pt]{article}
\usepackage[paperwidth=210mm, paperheight=11in]{geometry}
\usepackage{amsfonts,amsmath,amsthm,amssymb,amscd}
\usepackage{newpxtext,newpxmath} 
\usepackage{natbib}
\usepackage{tikz} 
\usetikzlibrary{arrows.meta}
\usepackage[nobottomtitles*]{titlesec}
\titleformat{\section}{\Large\bfseries\raggedright}{\thesection}{1em}{}
\titleformat{\subsection}{\large\bfseries\raggedright}{\thesubsection}{1em}{}
\titlespacing*{\section}{0pt}{*3}{.9ex}[\fill]
 
\usepackage{needspace} 
\usepackage{microtype} 
        \usepackage{graphicx}
\usepackage[colorlinks,linkcolor=blue,citecolor=blue]{hyperref} 

\usepackage{xcolor}

\newcommand{\reals}{\mathbb{R}} 
 
\newcommand{\T}{^{\top}} 
\newcommand{\0}{{\mathbf0}}
\newcommand{\1}{{\mathbf1}}

\def\endproof{\hfill\strut\nobreak\hfill\tombstone\par\smallbreak}
\def\tombstone{{\large$\square$}}

\newdimen\einr\einr1.8em
\newdimen\rmeinr\rmeinr1.8em
\newdimen\tmp \newcommand{\abs}[1]{\par\hangafter=1\hangindent=\einr
  \noindent\hbox to\einr{#1\hfill}\ignorespaces} 
 
\newcommand\bullitem{\tmp\einr\einr\rmeinr\abs{\raise.17ex\hbox{\kern7pt\scriptsize$\bullet$}}\einr\tmp}

\parindent\einr
\newtheorem{theorem}{Theorem}
\newtheorem{lemma}{Lemma}
\newtheorem{corollary}{Corollary}
\newtheorem{proposition}{Proposition}
\newtheorem{conjecture}{Conjecture}
\theoremstyle{definition}

\title{Equilibrium Numbers in Non-Square Bimatrix Games}
\author{Constantin Ickstadt\footnote{FB 12 -- Institut f\"ur Mathematik, Goethe-Universit\"at,
  Robert-Mayer-Str.~10, 60325 Frankfurt am Main, Germany.
  Email: ickstadt@math.uni-frankfurt.de,
  theobald@math.uni-frankfurt.de} \and Thorsten Theobald$^*$
  \and Bernhard von Stengel\footnote{Department of
  Mathematics, London School of Economics, London WC2A 2AE,
  United Kingdom. Email: b.von-stengel@lse.ac.uk}}

\date{November 23, 2026}
\begin{document}
\maketitle
\begin{abstract}
Bimatrix games may have an exponential number of mixed Nash
equilibria if both dimensions of the game are allowed to
grow.
Bounds on their maximal number give structural insights that
have been used to construct hard-to-solve games.
We show new sharp or asymptotically sharp bounds on the
(polynomial) number of equilibria for generic games where
one dimension of the game is fixed and the number of
strategies of the other player grows.
These results go beyond the hitherto studied square games.
Our methods employ combinatorial properties of polytopes,
and recent obstructions that relate to the graph of those
polytopes.
For $n\ge5$, we construct $3\times n$ games that have all
$2n+1$ vertices of the best-response polytope as equilibrium
strategies, proved using a simple case of the 4-color
theorem for planar graphs.
Generic $4\times 5$ games are shown to have at most 17
equilibria, using computer calculations with existing
datasets for all combinatorial types of the relevant
polytopes.
For $d\times n$ games, we construct games where all but a
fraction of $O(1/n)$ of the maximum number of vertices are
equilibrium strategies.
\end{abstract}

\newpage
\setcounter{tocdepth}{3} \tableofcontents
\newpage

\section{Introduction}

Bimatrix games are finite two-player games in strategic
form, a basic model of game theory.
Its central solution concept is Nash equilibrium, a pair of
mutually optimal strategies that always exists if players
are allowed to randomize \citep{Nash1951}.

We present new results on the possible number of Nash
equilibria in a generic bimatrix game.
Games with non-generic payoffs can have infinite sets of
Nash equilibria, so we only consider generic games.
The main novelty is to go beyond the square games hitherto
studied in this context.

It is well known that the number of mixed Nash equilibria in
a bimatrix game can be exponential if both dimensions of the
game are allowed to grow.
The $n\times n$ coordination game where both players have the
identity matrix as a payoff matrix has $2^n-1$ mixed
equilibria.
\citet{QS-1994,quint-shubik-1997}
conjectured that this is the maximum number of equilibria
for a generic $n\times n$ game.
This conjecture is true for $n\le4$
\citep{keiding-1997, MP-1999},
was disproved for $n\ge 6$ \citep{vS1999},
and recently answered positively for $n=5$ \citep{ITvS25}.

This paper studies equilibrium numbers beyond the
Quint-Shubik conjecture (now settled in full) by
considering $d\times n$ games that may not be square, that
is, for $d\ne n$.
We keep the number $d$ of rows fixed and construct games
with large numbers of equilibria as the number $n$ of
columns grows.
We choose $d$ (rather than~$m$) for the number of rows of
the game, which is the usual notation for a geometric
dimension that is kept fixed while $n$ grows.

\subsection{Our results and methods}
\label{s-results}

We show tight bounds on equilibrium numbers for $3\times n$
and $4\times 5$ games, and asymptotically tight bounds
(in~$n$) for fixed~$d$ for $d\times n$ games.

Our methods extend the use of \textit{polytopes} for
constructing games with many equilibria. 
Equilibria of bimatrix games correspond to certain 
complementary pairs of vertices of best-response
polytopes $P$ and $Q$ derived from the payoff matrices
(see (\ref{PQ}) in the preliminary Section~\ref{s-prelim}).
For a $d\times n$ game, the polytope $P$ has dimension
$d$ (also called a $d$-polytope), and the polytope $Q$
has dimension~$n$. 
Both polytopes have the same number $d+n$ of inequalities.
Each inequality corresponds to a pure strategy of a player. 
When the inequality is tight (holds as an equality), it
defines a facet of the polytope, \textit{labeled} with its
pure strategy.
An equilibrium is a completely labeled vertex pair. 

For the nondegenerate games that we consider, every polytope
vertex belongs to at most one equilibrium.
The number of vertices is therefore an upper bound on the
number of equilibria of the game.
Usually, for $d\le n$, the polytope $P$ of smaller
dimension has a smaller number of vertices, which therefore
provides a tighter bound on the number of equilibria.

In Section~\ref{s-23} we prove the new result that $3\times n$
games for $n=3$ and $n\ge 5$ have the maximal number $2n+1$
of Nash equilibria that use \textit{all} vertices of the
polytope~$P$.
(However, $3\times 4$ games have at most 7 Nash equilibria.)
A necessary and sufficient condition is that all facets of
$P$ are even-sided polygons.
It relates to the 3-colorability of
triangulated even-degree planar graphs, a simple case of the
4-color theorem.

For this result, two methods come into play:
Adjacent equilibrium vertices belong to equilibria of
opposite \textit{index} \citep{Shapley1974}, so that an
odd-sided polygon as face necessarily misses at least one
equilibrium vertex.
Secondly, it suffices to label the polytope $P$ only with
$d$, here three, rather than $d+n$ many labels, with a
simple construction of the high-dimensional $n$-polytope $Q$.
The corresponding concept of \textit{unit-vector games}
\citep{SvS2016} allows the construction of games from a
single labeled polytope rather than a pair of polytopes.

In Section~\ref{s-4}, we consider $4\times 4$ and
$4\times 5$ games.
As one new result for square games, we characterize the
combinatorial types of polytopes that give rise to $4\times
4$ games with 15 Nash equilibria.
We show that $4\times 5$ games have the surprisingly low
number of at most 17 Nash equilibria.
Both results use exhaustive computer calculations with
existing data sets of all combinatorial types of simple
4-polytopes \citep{firsching-2017} and 5-polytopes
\citep{fmm-2013} with 9 facets.
These are preprocessed with the help of the recent
\textit{stable-set bound} \citep{ITvS25} for equilibrium
numbers of polytopes, which substantially reduces the number
of polytopes to be checked for their 9! possible label
permutations.

Section~\ref{s-cyc} considers constructions of games based
on \textit{dual cyclic polytopes}, with largely 
combinatorial considerations.
Dual cyclic polytopes have the maximum number of vertices for a
given dimension and number of facets.
Their vertex-facet incidences are described by highly
regular bit patterns that fulfill the so-called
\textit{Gale evenness} condition 
(with a natural ``cyclic'' symmetry in even dimension).
Using these polytopes for $P$ and~$Q$ with a pairwise-swap
permutation of the labels of~$Q$ (see (\ref{swap}) and
(\ref{pi}) below) creates games with a large number of
equilibria.
This construction was first stated by \citet{vS1999} to
create $n\times n$ games with about $2.414^n$ equilibria (of
the maximum number of about $2.6^n$ vertices).
In Subsections~\ref{s-4xn} and~\ref{s-6} we show that for
$d\times n$ games derived from dual cyclic polytopes, the
number of equilibria obtained are tight for $d=4$ and $d=6$
and even~$n$.
These results use a known \textit{simplex} bound that
triangles and tetrahedra as polytope faces have at most two
equilibrium vertices.

In Subsection~\ref{s-drows} we show that all but $O(1/n)$ of
all vertices are equilibria of the $d\times n$ games
obtained from these dual cyclic polytopes.
Because they have the maximum numbers of vertices, this
is asymptotically tight for all $d\times n$ games.
Subsection~\ref{s-compute} reports the computationally
verified optimality of the swap permutation for
cyclic-polytope $8\times 8$ games, and related observations.
In Subsection~\ref{s-discuss} we note the
high precision required for the game payoffs to represent
the cyclic polytope structure correctly.
This ``numerical brittleness'' makes these games unlikely to
arise in practice.

\subsection{Relevance for equilibrium computation}
\label{s-equicomp}

Our results on equilibrium numbers are important for
equilibrium-finding algorithms in several respects.
The games with a large number of equilibria that we
construct provide test cases for algorithms that find
all Nash equilibria, and show how large their output can be.

In this paper we construct games where all or nearly all
vertices of a best-response polytope are equilibrium
strategies.
This shows that for general bimatrix games, finding all Nash
equilibria cannot be done faster than the \textit{lrsNash}
algorithm \citep{ARSvS} that enumerates all vertices of $P$
and matches complementary vertices of~$Q$.

The running time of the \textit{lrsNash} algorithm is
bounded by the maximum number of vertices of a polytope
(times a polynomial factor for solving the linear equations
that define a vertex).
This makes it the best method for finding all Nash
equilibria of bimatrix games.
An alternative is \textit{support enumeration} \citep{DK1991}.
In order to find the equilibria $(x,y)$ of a nondegenerate
$d\times n$ game, support enumeration tests all $d$-sized
subsets of $[d+n]$ as possible supports of~$y$ and unplayed
strategies of~$x$.
With $d$ fixed, this requires $O(n^d)$ steps
(times the same polynomial factor for solving the linear
equations for $x$ and $y$ as for vertex enumeration). 
In contrast, the \textit{Upper Bound Theorem} for polytopes
states that asymptotically, as a function of~$n$, the number
of vertices of $P$ is $O(n^{\lfloor d/2\rfloor})$
\citep{Mul1994}.
Both running times are polynomial in $n$, but the latter is
clearly better, in particular for larger~$d$.

A second relationship to equilibrium computation is the use
of dual cyclic polytopes for constructing hard-to-solve
bimatrix games for the algorithm by \citet{LH} that
finds one Nash equilibrium.
As shown by \citet{SvS2004, SvS2006}, this algorithm may
take exponentially many steps, even when including support
enumeration \citep{SvS2016}.
Because finding one Nash equilibrium in a bimatrix game is
PPAD-hard \citep{ChenDeng2006,CDT2009,DGP2009},
these exponential running times are sometimes associated
with a computational ``infeasibility'' of finding a Nash
equilibrium.
However, we discuss in Section~\ref{s-discuss} some caveats
to this claim because of the ``numerical brittleness'' of
cyclic polytopes.

The use of dual cyclic polytopes for the construction of
bimatrix games with large numbers of equilibria in
\cite{vS1999}, extended here, and long steps of the
Lemke-Howson algorithm (to find a unique equilibrium) is not
surprising, because the algorithm traverses polytope
vertices that are ``almost'' equilibria, where the
equilibrium condition holds for all but one of the pure
strategies.

\subsection{Related work}
\label{s-related}

\citet{mangasarian1964} described equilibria of bimatrix games 
in terms of polyhedra; see also \cite{von-stengel-2002, vS21}.
\textit{Random} bimatrix games (under certain distributional
assumptions for their payoffs) have an expected exponential
number of Nash equilibria \citep{McLennanBerg2005}.

The more restrictive two-player rank-1 games were introduced by
\citet{kannan-theobald-2010}.
They generalize zero-sum games in that the sum of the
players' two payoff matrices is allowed to be a matrix of
rank one (rather than zero). 
One Nash equilibrium of a rank-1 game can be found in
polynomial time, but these games can have an exponential
number of equilibria \citep{vS2012,AGMSvS2021}.

The number of Nash equilibria in $N$-player games has been
studied by \citet{mckelvey-mclennan-1997} and recently
\citet{HertlingVujic}.
For semi-definite games, equilibrium numbers were considered
by \citet{semidef-games} and \cite{semidef-network-games},
which are relevant for quantum games
\citep{bostanci-watrous-2021}.
The Quint-Shubik bound does hold for tropical games
\citep{agm-2023}.

\section{Background}
\label{s-prelim}

We review the concept of best-response polytopes in
Subsection~\ref{s-bima}, and the Upper Bound Theorem
for polytopes in Subsection~\ref{s-ubt}.
There we also explain the construction of dual cyclic
polytopes that achieve the upper bound, and the
Gale evenness condition that characterizes their vertices
via certain bit patterns.
For more detailed concepts about polytopes see
\cite{gruenbaum-2003} and \cite{ziegler-1995}.

\subsection{Bimatrix games and best-response polytopes}
\label{s-bima}

We use the notation for bimatrix games from \cite{ITvS25}.
All matrices have real entries.
The transpose of a matrix $B$ is denoted by $B\T$.
All vectors are column vectors.
Vectors and scalars are treated as matrices.
The $i$th component of a vector~$x$ is denoted by $x_i$\,.
The all-zero vector is denoted by~$\0$ and the all-one
vector by~$\1$, their dimension depending on the context.
Inequalities between vectors such as $x\ge\0$ 
hold for all components.
We let $[k]=\{1,\ldots,k\}$ for any positive integer~$k$.

A $d \times n$ \textit{bimatrix game} is a pair of $d\times
n$ matrices $(A,B)$.
The $d$ rows and $n$ columns are the \textit{pure strategies} of
player 1 and~2, respectively.
The players simultaneously each choose a pure strategy,
with the resulting entry of $A$ as payoff to player~1 and of
$B$ to player~2.
A \textit{mixed strategy} is a probability vector over the
player's pure strategies.
Players are interested in maximizing their expected payoff,
given by $x\T Ay$ to player~1 and $x\T By$ to player~2 if
their mixed strategies are $x$ and~$y$.
A mixed strategy is a \textit{best response} to the other
player's mixed strategy if it maximizes the player's
expected payoff over all possible mixed strategies.
A \emph{Nash equilibrium} is a mixed-strategy pair of mutual
best responses.

By adding a constant to each payoff if necessary, the
polyhedra
\begin{equation}
\label{PQ}
\arraycolsep0pt
\begin{array}{rllrlrl}
  P&{}=\{\,x\in {\mathbb R}^d&{}\mid {}&x&{}\ge\hbox{\bf0},~& B^{\top} x&{}\le \hbox{\bf1}\},\\
  Q&{}=\{\,y\in {\mathbb R}^n&{}\mid {}&Ay&{}\le \hbox{\bf1},~& y&{}\ge\bf 0\}\,.
\end{array}
\end{equation}
are bounded, called \textit{best-response polytopes}, of
dimension $d$ and $n$, respectively.

Both $P$ and $Q$ in (\ref{PQ}) are defined by $d+n$
inequalities.
Each inequality in (\ref{PQ}) is given a \textit{label}
in $[d+n]$.
The first $d$ labels refer to the pure strategies of
player~1 and the last $n$ labels to the pure strategies of
player~2.
A point $x$ in $P$ or $y$ in $Q$ \textit{has label} $k$ if the
corresponding inequality is tight.
That is, if $k\in [d]$ then $x$ has label~$k$ if $x_k=0$,
and $y$ has label~$k$ if $(Ay)_k=1$.
Similarly, if $k=d+j$ for $j\in [n]$ then $x$ has label~$k$
if $(B\T x)_j=1$, and $y$ has label~$k$ if $y_j=0$.
A pair $(x,y)$ in $P\times Q$ is called an
\textit{equilibrium} if it is \textit{completely labeled},
that is, each $k$ in $[d+n]$ appears as a label of $x$ or~$y$.
This includes the \textit{artificial equilibrium} $(\0,\0)$.
Any other equilibrium $(x,y)$ defines a Nash equilibrium 
of $(A,B)$ when re-scaling $x$ and $y$ as probability
vectors, and every Nash equilibrium is obtained in this way.
The complete set of labels represents the equilibrium
condition that every pure strategy is played with
probability zero or is a best response.

The game is assumed to be \textit{nondegenerate}, that is,
no mixed strategy has more pure best responses than the size
of its support \citep{Shapley1974,vS21,vS2022}, which holds generically. 
Equivalently, no $x$ in $P$ has more than $d$ labels and no
$y$ in $Q$ has more than $n$ labels.
Then $P$ and $Q$ are \textit{simple polytopes} (no more than
dimension many inequalities are tight).
This implies that if $(x,y)$ is an equilibrium, then $x$ is
a vertex of $P$ and $y$ is vertex of~$Q$, and every label is
either a label of $x$ or of $y$ but not of both.
Furthermore, every tight inequality in (\ref{PQ}) defines a
facet of $P$ or $Q$ (assuming the inequality is not
redundant), each with its respective label.
The number of vertices of $P$ (and of $Q$) is therefore a
bound on the number of Nash equilibria of $(A,B)$.

The polytopes $P$ and $Q$ in (\ref{PQ}) have the special
form that they sit in the nonnegative orthant. 
That is, both have $\0$ as a vertex and adjacent
nonnegativity constraints $x\ge\0$ as $d$ inequalities for
$P$ with labels $1,\ldots,d$, and $y\ge\0$ as $n$ inequalities
for $Q$ with labels $d+1,\ldots,d+n$.
It is possible to transform \textit{any} pair of simple
polytopes with facet labels in $[d+n]$ and a completely
labeled vertex pair $(x',y')$ into the shape (\ref{PQ}) via
a suitable affine transformation and a permutation of the
labels to map $(x',y')$ to the artificial equilibrium
$(\0,\0)$ \citep[Prop.~2.1]{vS1999}.
Bimatrix games can therefore be constructed from suitably
labeled simple polytopes of dimensions $d$ and $n$ with
$d+n$ facets.
In addition, all that matters for completely labeled vertex
pairs is the \textit{combinatorial structure} of these
polytopes, that is, which vertices lie on which facets.

\subsection{Cyclic polytopes and the Upper Bound Theorem for polytopes}
\label{s-ubt}

For a given dimension and number of facets,
simple polytopes that are \textit{dual-neighborly} have the
largest number of vertices.
The best-known instances of these polytopes are the dual
\textit{cyclic polytopes}, obtained as follows.
Consider $N$ points on the \textit{moment curve} in
dimension~$d$, given by all points
$\mu(t)=(t,t^2,\ldots,t^d)$ for $t\in\reals$.
The $N$ points are $\mu(t_i)$ for $t_1<t_2<\cdots<t_N$.
The cyclic polytope $C_d(N)$ is the convex hull of these
points, and its face structure does not depend on the
specific choices of the parameters~$t_i$.
Its dual (or polar) polytope $C_d(N)^\Delta$ is obtained 
by considering some interior point $\overline\mu$ of
$C_d(N)$ and letting
\begin{equation}
\label{cyc}
C_d(N)^\Delta = \{z \in\reals^d\mid
(\mu(t_i)-\overline \mu)\T z\le 1 
\hbox{ for }i\in[N]\}\,. 
\end{equation}
The cyclic polytope $C_d(N)$ is simplicial, that is, all its
facets are simplices, given by convex hulls of suitable 
$d$-element sets of its vertices $\mu(t_i)$.
It is \textit{neighborly}, which means that every vertex set
of size $\lfloor d/2\rfloor$ defines a face of the polytope.

These properties translate to the dual polytope
$C_d(N)^\Delta$ as follows.
Each vertex $\mu(t_i)$ of $C_d(N)$ corresponds to a facet of
$C_d(N)^\Delta$, via the respective tight inequality
in~(\ref{cyc}).
The polytope $C_d(N)^\Delta$ is simple because $C_d(N)$ is
simplicial.
It is dual-neighborly, that is, any $\lfloor d/2\rfloor$ of
its facets have a nonempty intersection.
Every vertex of $C_d(N)^\Delta$ corresponds to a facet
of~$C_d(N)$.
The neighborliness of~$C_d(N)$ (having the largest possible
number of low-dimensional faces) implies that $C_d(N)$ has
the largest possible number of facets; this is the essence
of the Upper Bound Theorem for polytopes, see
Theorem~\ref{UBT}.

The combinatorial structure of $C_d(N)$ is completely known
and the same for any parameters $t_1<\cdots<t_N$ when
written down in that order.
Encode a $d$-element set $S$ of vertices via a
\textit{bitstring} $b=b_1b_2\cdots b_N$ where $b_i=1$ if 
$\mu(t_i)\in S$ and $b_i=0$ otherwise
(so $b$ has $d$ 1's and $N-d$ 0's).
There is a unique hyperplane through the points in $S$,
which does not contain any further points of the moment
curve (otherwise the hyperplane would define a polynomial of
degree~$d$ with more than $d$ roots).
This hyperplane defines a facet of $C_d(N)$ if and only if 
all the other vertices $\mu(t_j)$ where $b_j=0$ are on the
same side of the hyperplane.
This holds if and only if $b$ fulfills the central
\textit{Gale evenness} condition \citep{gale1963},
which states that all substrings of
$b$ of the form $01^t0$ have an even number of 1's,
that is, $t$ is even ($1^t$ is a string of $t$ 1's);
the reason is that the moment curve always crosses the
hyperplane, so that an odd number of 1's would create
vertices that correspond to the two 0's on opposite sides of
the hyperplane, which therefore cannot be a facet.

The number $g(d,n)$ of bitstrings $b$ of length $d+n$ with
$d$ 1's and $n$ 0's that fulfill Gale evenness is determined
as follows.
Suppose $d$ is even, $d=2\ell$.
Then $b$ is composed either of $\ell$ strings 11 (pairs of 1's) and $n$ 0's,
with $\binom{\ell+n}{\ell}$ many possibilities,
or (with an odd number of 1's at both ends),
a 1 at each end of~$b$ with 
$\ell-1$ strings 11 and $n$ 0's in the middle,
which are $\binom{\ell-1+n}{\ell-1}$ many possibilities.
The total number is therefore
\begin{equation}
\label{GE}
g(2\ell,n)=
\binom{\ell+n}{\ell} + \binom{\ell-1+n}{\ell-1}
= \binom{\ell+n}{\ell}\left(1+\frac\ell{\ell+n}\right). 
\end{equation}
For even $d$, Gale evenness is preserved under cyclic
shifts of the bitstrings.

If $d$ is odd, $d=2\ell+1$, then the bitstring $b$ contains
an odd number of 1's and therefore has a single 1 at one of
its two ends.
The remaining string is composed of $\ell$ substrings 11 and
$n$ 0's, with $\binom{\ell+n}{\ell}$ many possibilities.
Hence,
\begin{equation}
\label{GEodd}
g(2\ell+1,n)=2 \binom{\ell+n}{\ell}.
\end{equation}

\begin{theorem}
[Upper Bound Theorem for polytopes \citep{McMullen1970}]
\label{UBT}
A simplicial $d$-polytope with $N$ vertices has
at most $g(d,N-d)$ many facets, and achieves this bound if
and only if it is neighborly.
A simple $d$-polytope with $N$ facets has at
most $g(d,N-d)$ many vertices, and achieves this bound if
and only if it is dual-neigh\-borly.
\end{theorem}

One polytope vertex in (\ref{PQ}) is $\0$, part of the
artificial equilibrium $(\0,\0)$, which does not represent a
mixed strategy.
Every other vertex may potentially represent an equilibrium
strategy, and is unique for that equilibrium.
The resulting upper bound on equilibrium numbers is known
\citep{keiding-1997}.

\begin{corollary}
A nondegenerate $d\times n$ bimatrix game
has at most $g(d,n)-1$ many Nash equilibria.
\end{corollary}

We will use later the combinatorial description via Gale
evenness strings of the vertices of the dual cyclic polytope
$C_d(N)^\Delta$.

\section{Games with two or three rows}
\label{s-23}

In this and the following sections, we study bounds on the
number of equilibria in nondegenerate bimatrix games
$(A,B)$ with best-response polytopes $P$ and $Q$ in
(\ref{PQ}).
If $(x,y)\in P\times Q$ and $(x,y)$ is an equilibrium,
then we call $x$ an \textit{equilibrium vertex} of~$P$
and $y$ its (unique) \textit{partner}, and vice versa.

The following constraints apply in general.

\begin{proposition}
\label{p-index}
In a nondegenerate bimatrix game $(A,B)$:

(a) The number of equilibria, including the artificial
equilibrium $(\0,\0)$, is even.

(b) Each equilibrium has \textit{index} $+1$ or $-1$.
Half of the equilibria have index $+1$ and
half (including the artificial equilibrium) have index $-1$.

(c) Any two equilibrium vertices of $P$ that are connected by
an edge belong to equilibria of opposite index.
\end{proposition}

Condition (a) is the \textit{parity argument} for equilibria 
(which leads to the complexity class PPAD, see \cite{Papa1994}),
because equilibria are endpoints of Lemke-Howson paths
\citep{LH}.

In (b), the index of an equilibrium $(x,y)\in P\times Q$ is
the sign of the determinant of the rows of the tight
inequalities in (\ref{PQ}) when written down in the order of
their labels, multiplied by $-1$ if $d+n$ is even to
normalize the index of the artificial equilibrium to $-1$.
\cite{Shapley1974} showed that the endpoints of
Lemke-Howson paths have opposite index
(for a compact proof see \cite{vS21}). This easily implies~(c) \citep[Lemma~5]{ITvS25}.

The \textit{graph} of a polytope is defined by its vertices and
edges (seen as unordered vertex pairs).
Proposition~\ref{p-index}(c) implies that for equilibria of
the same index, their equilibrium vertices are non-adjacent
in the graph of~$P$ (and analogously in the graph of $Q$).
This creates powerful \textit{obstructions} for equilibrium
numbers via the graphs of the polytopes (a stable set, also
called independent set, is a set of vertices no two of which
share an edge):

\begin{theorem}
[Stable-Set Bound, \citet{ITvS25}]
\label{t-ssb}
Let $G$ be the graph of $P$ and $S_1$ and $S_2$ be two
disjoint stable sets of~$G$ of equal size $|S_1|$, and let this
size be maximal among all such pairs.
Then the game has at most $2|S_1|$ equilibria.
\end{theorem}

\begin{lemma}
[Odd-Polygon Bound]
\label{l-odd}
If $P$ or $Q$ has an odd-sided polygon as a two-dimensional
face, then at least one vertex of that face is not an equilibrium
vertex.
\end{lemma}

\begin{lemma}
[Simplex Bound]
\label{l-simplex}
If $P$ or $Q$ has a simplex as a face (a clique in the graph
of the simple polytope), then at most two vertices of that
simplex are equilibrium vertices.
\end{lemma}

Lemma~\ref{l-simplex} can be proved without using the
equilibrium index, and has been used for triangles by
\citet{keiding-1997}.

We now consider $d\times n$ games for small $d$.
If $d=2$ then $P$ is a polygon with (at most) $n+2$
vertices, with an even number of equilibria.
This implies the following proposition, which is known
and where it is easy to construct the respective games
\citep{bgt-1989,QS-1994}.

\begin{proposition}
The maximum number of Nash equilibria of a nondegenerate
$2\times n$ game is $n+1$ if $n$ is even and $n$ if $n$ is odd.
\end{proposition}

Games with three rows, $d=3$, are more interesting.
Then $P$ is a 3-polytope (three-dimensional) and trivially
dual-neighborly because $\lfloor d/2\rfloor=1$
(any single facet has nonempty intersection).
By Theorem~\ref{UBT} and (\ref{GEodd}) where $\ell=1$, every
simple 3-polytope $P$ with $3+n$ facets has
at most $2n+2$ vertices.
This bound holds with equality by the Euler formula
\begin{equation}
\label{euler}
V-E+F=2
\end{equation}
for a 3-polytope $P$ with $V$ vertices, $E$
edges, and $F$ facets: Because $P$ is simple, every vertex
belongs to three edges and every edge has two vertices as
endpoints, which implies $E=\frac32 V$.
With $F=3+n$, (\ref{euler}) then gives $V=2n+2$.

As observed by \citet[Lemma 4.9]{QS-1994}, this gives an
upper bound of $2n+1$ Nash equilibria for a nondegenerate
$3\times n$ game.
We will construct $3\times n$ games that indeed have
$2n+2$ equilibria (including the artificial equilibrium)
for $n=3$ and all $n\ge 5$.

If $n=3$, then there are two polytopes with six facets, the
cube and the dual cyclic polytope $C_3(6)^\Delta$, shown
here with their Schlegel diagrams (projections to a facet):
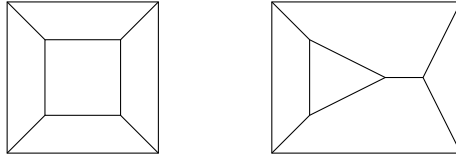
\[
\begin{tikzpicture}[scale=.5,
   >={Stealth[inset=1pt,length=7pt,angle'=30]}]
\draw[] (0,0) -- (4,0) -- (4,4) -- (0,4) -- cycle;
\draw[] (0,0) -- (1,1);
\draw[] (4,0) -- (3,1);
\draw[] (4,4) -- (3,3);
\draw[] (0,4) -- (1,3);
\draw[] (1,1) -- (3,1) -- (3,3) -- (1,3) -- cycle;
\begin{scope}[shift = {(7,0)}]
\draw[] (0,0) -- (5,0) -- (5,4) -- (0,4) -- cycle;
\draw[] (0,0) -- (1,1) -- (3,2) -- (4,2) -- (5,0);
\draw[] (0,4) -- (1,3) -- (3,2);
\draw[] (1,3) -- (1,1);
\draw[] (4,2) -- (5,4);
\end{scope}
\end{tikzpicture}
\]
The polytope $C_3(6)^\Delta$ has pentagons and triangles as
faces, so not all of its eight vertices can be equilibrium
vertices by Lemma~\ref{l-odd}.
The cube arises from the coordination game where both
players have the $3\times 3$ identity matrix as payoffs.
This game has seven Nash equilibria, so all vertices of the
cube are equilibria (including the artificial equilibrium).

A $3\times 4$ game can also have seven equilibria (by
adding a dominated strategy as column), but no more.

\begin{proposition}
\label{p-34}
A nondegenerate $3\times 4$ game has at most seven Nash
equilibria.
\end{proposition}

\proof
Assume that the polytope $P$ has seven facets.
Then $P$ has ten vertices.
In order for all of them to be equilibrium vertices,
all facets of $P$ have to be even-sided polygons
by Lemma~\ref{l-odd}. 
The only way to satisfy the Euler formula (\ref{euler}) 
is if these are six quadrilaterals and one hexagon
(with $(6\times 4+6)/3$ vertices).
However, there is no such 3-polytope.

Alternatively, the following pictures of the five
simple 3-polytopes with seven facets from \cite{GS-1967}
\[
\begin{tikzpicture}
  \coordinate (a) at (0,1);
  \coordinate (b) at (0.951,0.309);
  \coordinate (c) at (0.588,-0.809);
  \coordinate (d) at (-0.588,-0.809);
  \coordinate (e) at (-0.951,0.309);
  \coordinate (f) at (0*0.6,1*0.6);
  \coordinate (g) at (0.951*0.6,0.309*0.6);
  \coordinate (h) at (0.588*0.6,-0.809*0.6);
  \coordinate (m) at (-0.588*0.6,-0.809*0.6);
 \coordinate (n) at (-0.951*0.6,0.309*0.6);
  
  \draw (a)--(b); \draw (b)--(c); \draw (c)--(d);
  \draw (d)--(e);\draw (e)--(a);
  \draw (a)--(f); \draw (b)--(g); \draw (c)--(h);
  \draw (d)--(m); \draw (e)--(n);
  \draw (f)--(g);\draw (g)--(h); \draw (h)--(m);
  \draw (m)--(n); \draw (n)--(f);
\end{tikzpicture}\hspace*{2ex}
\begin{tikzpicture}
  \coordinate (a) at (0,1.1);
  \coordinate (b) at (0.951,0.209);
  \coordinate (c) at (0.588,-0.809);
  \coordinate (d) at (-0.588,-0.809);
  \coordinate (e) at (-0.951,0.209);
  \coordinate (f) at (0*0.6,1*0.6);
  \coordinate (g) at (0.851*0.6,0.309*0.6);
  \coordinate (h) at (0*0.6,-0.209*0.6);
  \coordinate (m) at (0*0.6,-0.809*0.6);
 \coordinate (n) at (-0.851*0.6,0.309*0.6);
  
  \draw (a)--(b); \draw (b)--(c); \draw (c)--(d);
  \draw (d)--(e);\draw (e)--(a);
  \draw (a)--(f); \draw (b)--(g);
  \draw (d)--(m); \draw (e)--(n);
  \draw (f)--(g);\draw (g)--(h); \draw (c)--(m); \draw (h)--(m);
  \draw (n)--(f); \draw (h)--(n);
\end{tikzpicture}\hspace*{2ex}
\begin{tikzpicture}
  \coordinate (a) at (1,0);
  \coordinate (b) at (0.5,0.866);
  \coordinate (c) at (-0.5,0.866);
  \coordinate (d) at (-1,0);
  \coordinate (e) at (-0.5,-0.866);
  \coordinate (f) at (0.5,-0.866);
  \coordinate (g) at (1*0.5,0*0.5);
  \coordinate (h) at (-1*0.5,0*0.5);
  \coordinate (m) at (-0.5*0.5,-0.866*0.5);
  \coordinate (n) at (0.5*0.5,-0.866*0.5);
   
  \draw (a)--(b); \draw (b)--(c); \draw (c)--(d);
  \draw (d)--(e);\draw (e)--(f); \draw (f)--(a);
  \draw (h)--(m); \draw (n)--(m); \draw (n)--(g);
  \draw (g)--(b); \draw (h)--(c);
  \draw (g)--(a);
  \draw (h)--(d);
  \draw (m)--(e);  
  \draw (n)--(f);
\end{tikzpicture}\hspace*{2ex}
\begin{tikzpicture}
  \coordinate (a) at (1,0);
  \coordinate (b) at (0.5,0.866);
  \coordinate (c) at (-0.5,0.866);
  \coordinate (d) at (-1,0);
  \coordinate (e) at (-0.5,-0.866);
  \coordinate (f) at (0.5,-0.866);
  \coordinate (g) at (1*0.3,0.4*0.5);
  \coordinate (h) at (-1*0.3,0.4*0.5);
  \coordinate (m) at (0,0);
  \coordinate (n) at (0,-0.866*0.5);
   
  \draw (a)--(b); \draw (b)--(c); \draw (c)--(d);
  \draw (d)--(e);\draw (e)--(f); \draw (f)--(a);
  \draw (h)--(m); \draw (n)--(m); \draw (n)--(e);
  \draw (g)--(b); \draw (h)--(c);
  \draw (g)--(a);
  \draw (h)--(d);
  \draw (m)--(g);  
  \draw (n)--(f);
\end{tikzpicture}\hspace*{2ex}
\begin{tikzpicture}
  \coordinate (a) at (1,0);
  \coordinate (b) at (0.5,0.866);
  \coordinate (c) at (-0.5,0.866);
  \coordinate (d) at (-1,0);
  \coordinate (e) at (-0.5,-0.866);
  \coordinate (f) at (0.5,-0.866);
  \coordinate (g) at (-0.5,-0.3);
  \coordinate (h) at (-0.2,0);
  \coordinate (m) at (0.2,0);
  \coordinate (n) at (0.5,0.3);
   
  \draw (a)--(b); \draw (b)--(c); \draw (c)--(d);
  \draw (d)--(e);\draw (e)--(f); \draw (f)--(a);
  \draw (a)--(n); \draw (b)--(n);
  \draw (c)--(h);
  \draw (d)--(g); \draw(e)--(g);
  \draw (f)--(m);
  \draw (g)--(h)--(m)--(n);
\end{tikzpicture}
\]
show that all of them have at least one triangle or
pentagon as a facet and miss one equilibrium vertex. 
Because the number of equilibria is even, $P$ has at most
eight equilibria.
\endproof

By Lemma~\ref{l-odd} a necessary condition for all $2n+2$
vertices of $P$ being equilibrium vertices is that all
facets of $P$ are even-sided polygons.
If $n$ is odd, this is easily achieved by considering an
$(n+1)$-gon $T$ (a convex polygon with $n+1$ vertices) and
defining $P$ as the \textit{prism over}~$T$ given by
$P=T\times [0,1]$ (for example, the cube if $n=3$ and $T$ is
the unit square).
For such a prism, it is possible to immediately create the
payoff matrix $B$ in (\ref{PQ}) by drawing $T$ in the
nonnegative quadrant of $\reals^2$ with corner $(0,0)$.

For even $n\ge6$, a 3-polytope with $3+n$ even-sided
polygons as facets can be obtained as follows.
Consider a prism over an $(n-2)$-gon~$T$.
This prism has $n$ facets: the top and bottom facets
$T\times\{1\}$ and
$T\times\{0\}$, and
$n-2$ rectangular sides.
Consider some vertex $v$ on the top facet with corresponding
vertex~$u$ on the bottom facet, and truncate the prism by
cutting off~$v$ with three additional inequalities that
create three new quadrilateral facets as in this picture:
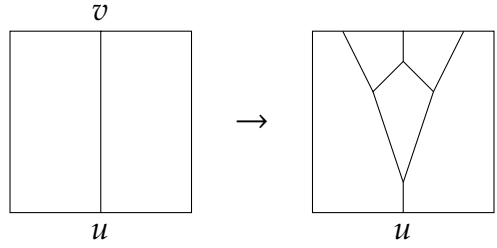
\[
\begin{tikzpicture}[scale=.4]
\coordinate (A) at (0,0);
  \coordinate (B) at (3,0);
  \coordinate (C) at (3,-6);
  \coordinate (D) at (0,-6);
  \coordinate (E) at (-3,-6);
  \coordinate (F) at (-3,0);
\draw[] (A) node[above]{$v$};
\draw[] (D) node[below]{$u$};
\draw[] (A) -- (B) -- (C) -- (E) -- (F) -- (A) -- (D); 
\draw[] (5,-3) node{$\to$};

\begin{scope}[shift={(10,0)}]
\coordinate (A) at (0,0);
  \coordinate (B) at (3,0);
  \coordinate (C) at (3,-6);
  \coordinate (D) at (0,-6);
  \coordinate (E) at (-3,-6);
  \coordinate (F) at (-3,0);
  \coordinate (G) at (-2,0);
  \coordinate (H) at (2,0);
  \coordinate (I) at (0,-1);
  \coordinate (J) at (-1,-2);
  \coordinate (K) at (1,-2);
  \coordinate (L) at (0,-5);
\draw[] (D) node[below]{$u$};
\draw[] (A) -- (B) -- (C) -- (E)
-- (F) -- (A) -- (I) -- (K) -- (L) -- (D); 
\draw[] (G) -- (J) -- (I);
\draw[] (H) -- (K);
\draw[] (J) -- (L);
\end{scope}
\end{tikzpicture}
\]
The top facet then gets two additional vertices, and the two
rectangular sides with common vertex~$u$ become hexagons.
The general shape is shown here, with the originally
rectangular sides numbered $1,\ldots,n-2$, the three new
quadrilaterals numbered $n-1,n,n+1$, the
unchanged bottom facet numbered $n+2$ 
and the new top facet numbered $n+3$:
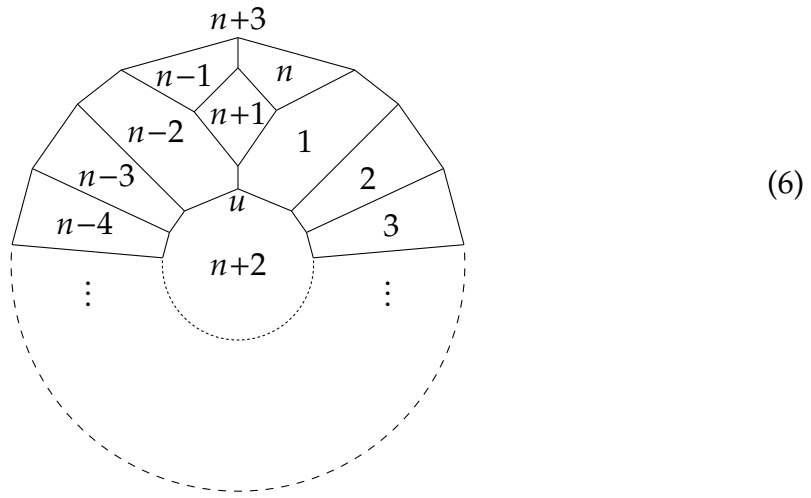
\begin{equation}
\label{top}
\lower20ex\hbox{
\begin{tikzpicture}[yscale=-1,
 ,>=latex, inner sep=0pt, outer sep=2pt, axis/.style={thin,gray}, vecteur/.style=
{->,thick,color=black,smooth}, courbe/.style={thick,color=#1,smooth}]
    \colorlet{darkblue}{blue!50!black};
    \foreach \t in {0,...,2} {
      \draw ({1.0*cos(\t *20 + 315)},{1.0*sin(\t *20 + 315)}) -- ({3*cos(\t *20 + 315)},{3*sin(\t *20 + 315)});
      \draw ({1.0*cos(- \t *20 + 225)},{1.0*sin(- \t *20 + 225)}) -- ({3*cos(- \t *20 + 225)},{3*sin(- \t *20 + 225)});
     }

    \draw ({1.0*cos(0 *20 + 315)},{1.0*sin(0 *20 + 315)}) -- ({1.0*cos(0 *20 + 270)},{1.0*sin(0 *20 + 270)}); 
   \draw ({1.0*cos(0 *20 + 225)},{1.0*sin(0 *20 + 225)}) -- ({1.0*cos(0 *20 + 270)},{1.0*sin(0 *20 + 270)});
    \draw ({1.0*cos(0 *20 + 270)},{1.0*sin(0 *20 + 270)}) -- ({1.3*cos(0 *20 + 270)},{1.3*sin(0 *20 + 270)});
    \draw ({1.3*cos(0 *20 + 270)},{1.3*sin(0 *20 + 270)}) -- ({2.1*cos(0 *20 + 284)},{2.1*sin(0 *20 + 284)});
    \draw ({1.3*cos(0 *20 + 270)},{1.3*sin(0 *20 + 270)}) -- ({2.1*cos(0 *20 + 254)},{2.1*sin(0 *20 + 254)});
    \draw ({2.6*cos(0 *20 + 270)},{2.6*sin(0 *20 + 270)}) -- ({2.1*cos(0 *20 + 284)},{2.1*sin(0 *20 + 284)});
    \draw ({2.6*cos(0 *20 + 270)},{2.6*sin(0 *20 + 270)}) -- ({2.1*cos(0 *20 + 254)},{2.1*sin(0 *20 + 254)});
    \draw ({2.6*cos(0 *20 + 270)},{2.6*sin(0 *20 + 270)}) -- ({3.0*cos(0 *20 + 270)},{3.0*sin(0 *20 + 270)}); 
    \draw ({3.0*cos(0 *20 + 225)},{3.0*sin(0 *20 + 225)}) -- ({3.0*cos(0 *20 + 239)},{3.0*sin(0 *20 + 239)}) 
        -- ({3.0*cos(0 *20 + 270)},{3.0*sin(0 *20 + 270)}); 
    \draw ({3.0*cos(0 *20 + 315)},{3.0*sin(0 *20 + 315)}) -- ({3.0*cos(0 *20 + 301)},{3.0*sin(0 *20 + 301)})
       -- ({3.0*cos(0 *20 + 270)},{3.0*sin(0 *20 + 270)});
    \draw ({2.1*cos(0 *20 + 254)},{2.1*sin(0 *20 + 254)}) -- ({3.0*cos(0 *20 + 239)},{3.0*sin(0 *20 + 239)});
    \draw ({2.1*cos(0 *20 + 284)},{2.1*sin(0 *20 + 284)}) -- ({3.0*cos(0 *20 + 301)},{3.0*sin(0 *20 + 301)});
    \foreach \t in {0,...,1} {
       \draw ({3*cos((\t+1) *20 + 315)},{3*sin((\t+1) *20 + 315)}) -- ({3*cos(\t *20 + 315)},{3*sin(\t *20 + 315)});
       \draw ({3*cos((-\t-1) *20 + 225)},{3*sin((-\t-1) *20 + 225)}) -- ({3*cos(-\t *20 + 225)},{3*sin(-\t *20 + 225)});
       \draw ({1.0*cos((\t+1) *20 + 315)},{1.0*sin((\t+1) *20 + 315)}) -- ({1.0*cos(\t *20 + 315)},{1.0*sin(\t *20 + 315)});
       \draw ({1.0*cos((-\t-1) *20 + 225)},{1.0*sin((-\t-1) *20 + 225)}) -- ({1.0*cos(-\t *20 + 225)},{1.0*sin(-\t *20 + 225)});
    }
    \foreach \t in {-23,...,23} {
      \draw ({1.0*cos(\t *4 +2 +90-1)},{1.0*sin(\t *4 +2 +90-1)}) --({1.0*cos(\t *4 +90-1)},{1.0*sin(\t *4 +90-1)});  
       \draw ({3.0*cos(\t *4 +2 +90-1)},{3.0*sin(\t *4 +2 +90-1)}) --({3.0*cos(\t *4 +90-1)},{3.0*sin(\t *4 +90-1)});  
    }  
  \node[] at (0,0) {$n{+}2$};
  \node[] at (0,-.8) {$u$};
  \node[] at ({2.0*cos(0 *20 + 172)},{2.0*sin(0 *20 + 172)}){$\vdots$};
  \node[] at ({2.1*cos(0 *20 + 195)},{2.1*sin(0 *20 + 195)}){$n{-}4$};
  \node[] at ({2.1*cos(0 *20 + 214)},{2.1*sin(0 *20 + 214)}){$n{-}3$};
  \node[] at ({2.0*cos(0 *20 + 237)},{2.1*sin(0 *20 + 237)}){$n{-}2$};
  \node[] at ({2.6*cos(0 *20 + 254)},{2.6*sin(0 *20 + 254)}){$n{-}1$};
  \node[] at ({2.6*cos(0 *20 + 284)},{2.6*sin(0 *20 + 284)}){$n$};
  \node[] at ({2.0*cos(0 *20 + 270)},{2.0*sin(0 *20 + 270)}){$n{+}1$};
  \node[] at ({2.1*cos(0 *20 + 295)},{2.0*sin(0 *20 + 305)}){$1$};
  \node[] at ({2.1*cos(0 *20 + 325)},{2.0*sin(0 *20 + 325)}){$2$};
  \node[] at ({2.1*cos(0 *20 + 345)},{2.0*sin(0 *20 + 345)}){$3$};
  \node[] at ({2.0*cos(0 *20 + 8)},{2.0*sin(0 *20 + 8)}){$\vdots$};
  \node[] at ({3.25*cos(0 *20 + 270)},{3.25*sin(0 *20 + 270)}){$n{+}3$};
\end{tikzpicture}}\end{equation}
For $n=6$ with $T$ as the unit square and $v=(1,1,1)$,
such a truncation of the unit cube is shown in (\ref{123})
below.
The respective inequalities give a matrix $B$ (after
multiplication by~9) and a game $(A,B)$ with
\begin{equation}
\arraycolsep .3em
\label{AB}
A=
\left(\,
\begin{matrix}
1 & 0 & 0 & 1 & 0 & 0\\
0 & 1 & 0 & 0 & 1 & 0\\
0 & 0 & 1 & 0 & 0 & 1\\
\end{matrix}
\,\right)\,,
\qquad
B=\left(\,
\begin{matrix}
9 & 0 & 0 & 2 & 4 & 4\\
0 & 9 & 0 & 4 & 2 & 4\\
0 & 0 & 9 & 4 & 4 & 2\\
\end{matrix}
\,\right)\,.
\end{equation}
Then $P$ has indeed 14 vertices and $(A,B)$ has 13 Nash
equilibria!

Remarkably, the condition that $P$ has only even-sided
polygons as facets is also \textit{sufficient} for all
vertices to be equilibrium vertices.
How could a suitable payoff matrix $A$ in (\ref{AB}), which
defines a high-dimensional $n$-polytope $Q$, be defined so
simply?
The answer is the concept of \textit{unit vector games},
first defined by \cite{balthasar2009} and studied in
\cite{SvS2016}.
They are about $d$-polytopes with facet labels \textit{only}
in~$[d]$, whose completely labeled vertices correspond to
the equilibria of a suitable game, called a unit vector
game.
The following proposition \citep[Prop.~1]{SvS2016} is
straightforward to show.
The unit vector $e_k$ has 1 in its $k$th component and 0
elsewhere.

\begin{proposition}
[\citet{SvS2016}]
\label{uvg}
Let $P'$ be a simple $d$-polytope,
$P'=\{x\in\reals^d\mid {x\ge\0},~{B\T x\le\1}\}$
for a $d\times n$ matrix $B$,
where every inequality has a label in $[d]$
and each inequality $x_i\ge0$ has label $i$ for $i\in[d]$.
Then the completely labeled vertices $x$ of $P'$ (those that
have all labels in $[d]$) correspond to the equilibria
$(x,y)$ of the game $(A,B)$ where the $j$th column of $A$ is
the unit vector $e_{l(j)}$ if the inequality $(B\T x)_j\le1$
has label $l(j)$, for $j\in[n]$.
\end{proposition}

The following picture shows the polytope $P'$ (projected to
its top face) for the matrix $B$ in (\ref{AB}) with 
facet labels 1, 2, or~3 as circled numbers. 
The facets and their labels define the
unit-vector columns of $A$ except for the facets indicent
to the origin~$\0$ that define the constraints $x\ge\0$ for~$P'$. 
The facets with labels 1 and 2 at~$u$ and the outer facet~3
correspond to the first three columns of~$A$, and the
remaining facets to the last three columns.
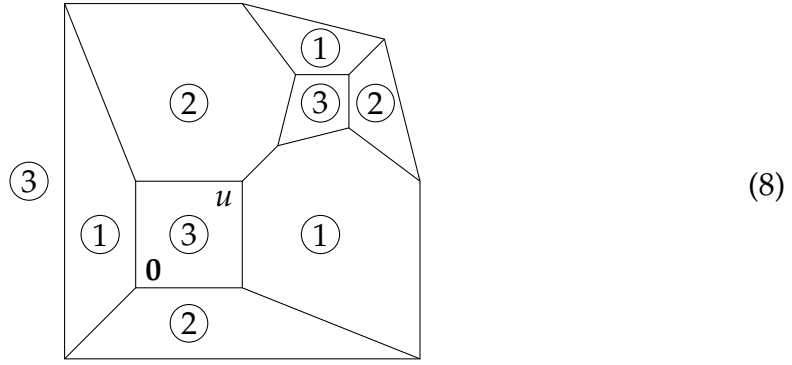
\begin{equation}
\label{123}
\lower11ex\hbox{\begin{tikzpicture}[scale=.47]
\coordinate (A) at (0,10);
  \coordinate (B) at (0,0);
  \coordinate (C) at (10,0);
  \coordinate (D) at (10,5);
  \coordinate (E) at (9,9);
  \coordinate (F) at (5,10);
  \coordinate (G) at (6.5,8);
  \coordinate (H) at (6,6);
  \coordinate (I) at (5,5);
  \coordinate (J) at (2,5);
  \coordinate (K) at (2,2);
  \coordinate (L) at (5,2);
  \coordinate (M) at (8,6.5);
  \coordinate (N) at (8,8);
\draw[] (F) -- (A) -- (B) -- (C) -- (D) -- (E)
-- (F) -- (G) -- (H) -- (I) -- (J) -- (K) -- (L) -- (I); 
\draw[] (H) -- (M) -- (N) -- (G);
\draw[] (A) -- (J);
\draw[] (B) -- (K);
\draw[] (C) -- (L);
\draw[] (D) -- (M);
\draw[] (E) -- (N);
\node[draw,circle,inner sep=.12em] at (3.5,3.5) {3};
\node[draw,circle,inner sep=.12em] at (3.5,1) {2};
\node[draw,circle,inner sep=.12em] at (1,3.5) {1};
\node[draw,circle,inner sep=.12em] at (7.2,3.5) {1};
\node[draw,circle,inner sep=.12em] at (7.2,7.2) {3};
\node[draw,circle,inner sep=.12em] at (3.5,7.2) {2};
\node[draw,circle,inner sep=.12em] at (7.2,8.75) {1};
\node[draw,circle,inner sep=.12em] at (8.75,7.2) {2};
\node[draw,circle,inner sep=.12em] at (-1,5) {3};
\node[] at (4.5,4.5) {$u$};
\node[] at (2.5,2.5) {$\0$};
\end{tikzpicture}}
\end{equation}

In the labeled polytope shown in (\ref{123}),
every vertex has all three labels and is therefore an
equilibrium vertex in the game (\ref{AB}).
This is possible because every facet has an even number of
neighboring facets, which is a necessary condition.
It is also sufficient, using an easy case of the 4-color
theorem for planar graphs.
We apply it to the dual polytope $P^\Delta$ of~$P$, which is
simplicial. 
The facets of $P$ become vertices of $P^\Delta$ with an even
number of neighbors, and the vertices of~$P$ become
triangular facets of~$P^\Delta$.

\begin{proposition}
\label{3colors}
A triangulated planar graph where every vertex has even
degree is 3-colorable.
\end{proposition}

See \cite{tsai2011} for history and a proof (via a rather
straightforward greedy algorithm; the color classes of the
3-coloring are in fact unique).

The 3-colorability of the graph of $P^\Delta$ means that
every vertex of $P$ is adjacent to three facets of~$P$ with
different colors.
With colors as labels, this means all vertices are
completely labeled.

\begin{corollary}
For $n=3$ or $n\ge 5$, there are nondegenerate $3\times n$
bimatrix games with $2n+1$ many Nash equilibria.
\end{corollary}

\section{Games with four rows and four or five columns}
\label{s-4}

In Subsection~\ref{s-4x4} we identify the three
combinatorial types of best-response polytopes for $4\times 4$ games
with the maximum number of 16 equilibria (including the
artificial equilibrium).
This includes the dual cyclic polytope, for which we present
a general equilibrium construction based on a specific
``swap'' labeling of the second polytope~$Q$.
In Subsection~\ref{s-4x5} we prove by exhaustive computer
search that $4\times 5$ games have at most 18 equilibria. 
That bound is again attained for dual cyclic polytopes,
among others.

\subsection{$4\times 4$ games and general games from cyclic polytopes}
\label{s-4x4}

A nondegenerate $4\times 4$ game has at most 15 Nash
equilibria, as shown by \citet{keiding-1997} and \citet{MP-1999}.
By Theorem~\ref{UBT} and (\ref{GE}), the polytope $P$ can
have up to 20 vertices.

The proof by \citet{keiding-1997} uses the 
classification in \citet{GS-1967} of all 37 combinatorial
types of simple 4-polytopes with up to 8 facets,
by finding suitably many disjoint triangles as facets that
reduce the number of equilibrium vertices by
Lemma~\ref{l-odd}.
In this section, we show that three different polytopes can
be used to obtain 16 equilibria.
One of them is the 4-dimensional cube that corresponds to
the coordination game $(A,B)$ with $A$ and $B$ as the
$4\times 4$ identity matrix.

Three of the simple 4-polytopes with 8 facets
in \cite{GS-1967} have 20 vertices and are dual-neighborly.
One of them is the dual cyclic polytope $C_4(8)^\Delta$.
In order to obtain 16 equilibria, both $P$ and $Q$ can be
chosen to be this polytope, but their facet labelings
have to be different.
Namely, assume that the facets of $P$ are labeled with
$1,\ldots,8$ in the order of the positions of the Gale
evenness bitstrings that represent the vertices, as per the
order of $t_1<\cdots<t_8$ in the definition of
$C_4(8)^\Delta$ in~(\ref{cyc}).
For the labels of $Q$, we use the following \textit{swap
permutation} that exchanges the odd and even positions
\begin{equation}
\label{swap}
(1,2,3,4,5,6,7,8)\mapsto (2,1,4,3,6,5,8,7).
\end{equation}
For Gale evenness strings of even length $d+n$, we define 
this swap permutation $\pi:[d+n]\to[d+n]$ generally as 
\begin{equation}
\label{pi}
\pi(i)=\begin{cases}
i+1 & \hbox{if $i$ is odd}\\
i-1 & \hbox{if $i$ is even}\\
\end{cases}
\quad
\hbox{for $i\in[d+n].$}
\end{equation}
The dual cyclic polytope $Q$ is labeled such that its $i$th
facet (corresponding to the $i$th position in a Gale
evenness bitstring) gets label~$\pi(i)$, for $i\in[d+n]$.

Remark:
This construction was used in \cite{vS1999} to create a
$6\times6$ game with 75 Nash equilibria, and $d\times d$ games
with about $(1+\sqrt 2)^d/\sqrt d$ many equilibria.
This is much larger than the $2^d$ equilibria for the
coordination game, but short of the bound of about
$2.6^d/\sqrt d$ from the Upper Bound Theorem, using
Stirling's formula for~(\ref{GE}).

The swap permutation $\pi$ for the labels of $Q$ has the
following effect.
Let $d$ and $n$ be even, $d=2\ell$ and $n=2k$.
Call a bitstring with $d$ 1's and $n$ 0's
$(\ell,k)$-\textit{swappable} if it is composed of substrings of the
form 0110, 11, or 00
(this requires $d$ and $n$ to be even).
If there are $i$ substrings of the form 0110, where
$0\le i\le \ell$, then there are $l-i$ strings 11 (to get
$2\ell$ 1's in the string) and $k-i$ strings 00 (to get $2k$
0's).
Therefore, the number of $(\ell,k)$-swappable strings is the
sum of multinomial coefficients
\begin{equation}
\label{sigma}
\sigma(\ell,k)=\sum_{i=0}^\ell\frac{(k+\ell-i)!}{i!\,(\ell-i)!\,(k-i)!} 
=\sum_{i=0}^\ell\binom \ell i \binom{k+\ell-i}{\ell} \,.
\end{equation}
Clearly, every swappable string fulfills Gale evenness
and is therefore a vertex of~$P$.
The complement of such a string (changing 0 to~1 and vice
versa) should represent a vertex of~$Q$ to obtain an
equilibrium, taking the labels of $Q$ into account.
That complemented string is composed of substrings 1001, 00, and 11,
respectively, which under the re-labeling (\ref{pi})
become 0110, 00, and 11, and therefore again fulfills Gale
evenness (it is then $(k,\ell)$-swappable).
Hence, every $(\ell,k)$-swappable string represents an
equilibrium vertex of~$P$.
In addition, the $(\ell-1,k-1)$-swappable strings preceded
by 10 and succeeded by 01 also have this property
(for example, if $\ell=k=2$, then these are the three
strings
$10\,1100\,01$,
$10\,0011\,01$,
and
$10\,0110\,01$). 
With the labeling (\ref{pi}) of~$Q$, these represent exactly
the equilibrium vertices of~$P$.
We have shown the following, a generalization of
\citet[Thm.~4.2]{vS1999} that considered the case $d=n$.

\begin{theorem}
\label{t-cyc}
Let $d=2\ell$ and $n=2k$.
There is a $d\times n$ game with
$\sigma(\ell,k)+\sigma(\ell-1,k-1)$ many equilibria.
\end{theorem}

For the $4\times 4$ game with $\ell=k=2$, we therefore have
$\sigma(2,2)+\sigma(1,1)=13+3=16$ equilibria.
The permutation in (\ref{swap}) for labeling $Q$ is not
unique -- it can be cyclically shifted and reversed, because
the set of Gale evenness strings has these symmetries.

Apart from the $4\times 4$ game with $P$ as the 4-cube, there is
another neighborly polytope called $N_8$ in 
\citet[p.~125]{gruenbaum-2003} that in its dual version
$N_8^\Delta$ gives rise to 16 equilibria if both $P$
and $Q$ are set to this polytope.
The 20 vertices of $N_8^\Delta$ are given as follows, by listing
the respective facets $1,\ldots,8$ that each vertex lies on;
we could also list them as bitstrings but they are not Gale
evenness strings.
\begin{equation}
\label{n8}
  \begin{array}{ccccc}
  1234, & 1237, & 1245, & 1256, & 1268, \\
  1278, & 1348, & 1378, & 1458, & 1568, \\
  2345, & 2356, & 2367, & 2678, & 3456, \\
  3467, & 3478, & 4567, & 4578, & 5678.
  \end{array}
\end{equation}
We consider the permutation 
\begin{equation}
\label{permn8}
(1,2,3,4,5,6,7,8)\mapsto (2,1,5,6,3,4,8,7) 
\end{equation}
as the corresponding labeling of $Q$ for the same polytope.
For example, applying this permutation to 1237 gives 2158,
which in ascending order of digits is 1258, whose complement
3467 is, like 1237, a vertex of $N_8^\Delta$.
Hence, these vertices form an equilibrium pair.

The following table lists the vertices in (\ref{n8}) after
applying the permutation (\ref{permn8}) and sorting the
digits, which are the facet labels of the vertices of~$Q$.
We underline those vertices $y$ which form a complement of
a vertex $x$ of $P$ in (\ref{n8}); this gives 16 equilibria
$(x,y)$.
\[
  \begin{array}{ccccc}
  \underline{1256}, & \underline{1258}, & \underline{1236}, & \underline{1234}, & 1247, \\
  \underline{1278}, & \underline{2567}, & 2578, & \underline{2367}, & \underline{2347}, \\
  1356, & \underline{1345}, & \underline{1458}, & \underline{1478}, & \underline{3456}, \\
  \underline{4568}, & \underline{5678}, & 3468, & \underline{3678}, & \underline{3478}.
  \end{array}
\]
We found (by testing all $8!$ possibilities) four permutations that give 16
equilibria, by replacing $(2,1,5,6,3,4,8,7)$ in (\ref{permn8}) 
by $(3,4,1,2,7,8,5,6)$,
$(6,5,8,7,2,1,4,3)$, or
$(7,8,4,3,6,5,1,2)$.

There is a third neighborly 4-polytope with 8 vertices
called $N_8^*$ in \citet[p.~125]{gruenbaum-2003}.
However, its dual polytope has a stable-set bound from
Theorem~\ref{t-ssb} of 14 and therefore cannot have 16
equilibria.
There are no further polytopes with a stable-set bound
of~16, and both polytopes have to be combinatorially
isomorphic to obtain 16 equilibria.

In summary, $4\times 4$ games with 16 equilibria can be
obtained in three combinatorially different ways from
polytopes $P$ and $Q$.
In all three cases, $P$ and $Q$ have the same combinatorial
structure, which is either that of the 4-cube, the dual cyclic
polytope $C_4(8)^\Delta$, or $N_8^\Delta$ from
\citet[p.~125]{gruenbaum-2003}.

\subsection{$4\times 5$ games}
\label{s-4x5}

For a $4\times 5$ game, $P$ and $Q$ have 9 facets.
According to Theorem~\ref{UBT}, $P$ has at most 27 and $Q$
at most 30 vertices.
We show that the number of equilibria is substantially
smaller than the upper bound of~26.

\begin{lemma}
\label{l-17}
There exists a nondegenerate $4 \times 5$ game 
with $17$ Nash equilibria.
\end{lemma}

\proof
The following construction is based on cyclic polytopes and
yields $18$ equilibria.
We let $P=C_4(9)^\Delta$ and $Q=C_5(9)^\Delta$.
The labels of $Q$ will be permuted with the permutation
\begin{equation}
\label{swap9}
(1,2,3,4,5,6,7,8,9)\mapsto (2,1,4,3,6,5,8,7,9)\,,
\end{equation}
which is similar to (\ref{swap}) except that label~9 stays
fixed.
The following is the list of the 27 Gale evenness strings
that represent the vertices of $P$, where we underline those
that have an equilibrium partner in $Q$ under the
permutation in (\ref{swap9}).
These are exactly the 13 $(2,2)$-swappable strings followed by
a terminal~digit 0, or an initial substring 10 followed by a
$(1,2)$-swappable string followed by a terminal digit~1;
the number of the latter strings is $\sigma(1,2)=5$.
\begin{equation}
\label{ge49}
\begin{array}{c@{\quad}c@{\quad}c@{\quad}c@{\quad}c}
    \underline{11\,11\,00\,00\,0} &
    \underline{11\,01\,10\,00\,0} &
    \underline{11\,00\,11\,00\,0} &
    \underline{11\,00\,01\,10\,0} &
    \underline{11\,00\,00\,11\,0} \\
               11\,00\,00\,01\,1 & 
               01\,11\,10\,00\,0 & 
    \underline{01\,10\,11\,00\,0} &
    \underline{01\,10\,01\,10\,0} &
    \underline{01\,10\,00\,11\,0} \\
               01\,10\,00\,01\,1 & 
    \underline{00\,11\,11\,00\,0} &
    \underline{00\,11\,01\,10\,0} &
    \underline{00\,11\,00\,11\,0} &
               00\,11\,00\,01\,1 \\
               00\,01\,11\,10\,0 & 
    \underline{00\,01\,10\,11\,0} &
               00\,01\,10\,01\,1 &
    \underline{00\,00\,11\,11\,0} &
               00\,00\,11\,01\,1 \\
               00\,00\,01\,11\,1 &
               11\,10\,00\,00\,1 &
    \underline{10\,11\,00\,00\,1} &
    \underline{10\,01\,10\,00\,1} &
    \underline{10\,00\,11\,00\,1} \\
    \underline{10\,00\,01\,10\,1} &
    \underline{10\,00\,00\,11\,1} 
\end{array}      
\end{equation}  
This shows that this polytope pair has $13+5=18$ equilibria and the
corresponding game therefore 17 Nash equilibria.
\endproof 

It turns out that this is also the maximum number.
This based on a computer analysis of existing lists of
all combinatorial types of simplicial 4-polytopes and
5-polytopes with 9 vertices, which represent in the same way
the combinatorial structure of their dual counterparts, the
respective simple polytopes with 9 facets.

For the simplicial $4$-polytopes with 9 vertices, we use the
data of \cite{firsching-2017}, with coordinates
for the respective polytopes, of which there are 1142
combinatorial types.

A classification of 5-polytopes with 9 facets was 
provided by \cite{fmm-2013}.
There are 47923 combinatorial types, out of which 322
are simple (see also \cite{firsching-2017}).
We use coordinate data for these 322 types, which was
provided to us by Moritz Firsching. 

\begin{theorem}
\label{45-17}
The maximum number of Nash equilibria in a nondegenerate $4
\times 5$ game is 17.
\end{theorem}

\proof
Because the number of equilibria is even,
it suffices to show that 20 or more equilibria are not possible.
Assume that there exists a non-degenerate $4 \times 5$ game with 
20 equilibria.
Then both simple polytopes must have a stable-set bound of
at least 20. 

Among all 1142 types of simple $4$-polytopes with 9 facets,
the maximum of the stable-set bounds according to
Theorem~\ref{t-ssb} has been found to be~20.
There are 39 types where the stable-set bound gives~20.
Among all the 322 types of simple $5$-polytopes with 9 facets, 
the maximum of the stable-set bounds is also 20.
There are 68 types with that bound of~20.

Hence, the 4-polytope $P$ is among the 39 types with
stable-set bound 20 and the 5-polytope $Q$
is among the 68 types with stable-set bound~20.

We have run a computer program which inspects the $39 \cdot 68$
combinations on all 9! relabelings.
From the output of the program, all the polytope pairs and
all the relabelings give at most 18 equilibria.
This proves the claim.
The data and code used are documented in \cite{complementary-data}.
\endproof

\section{Cyclic polytopes and games} 
\label{s-cyc}

In this section we prove tight and, for growing $n$,
asymptotically tight bounds on the equilibrium numbers of
$d\times n$ games based on dual cyclic polytopes.

\subsection{$4\times n$ games}
\label{s-4xn}

For $4\times n$ games, Theorem~\ref{UBT} shows that they can
have at most $\binom{n+2}2+\binom{n+1}1=(n^2+5n+4)/2$ many
equilibria.
By subtracting the artificial equilibrium, this agrees with
the bound of $(n^2+5n+2)/2$ many Nash equilibria 
by \cite{QS-1994}.
Their bound is based on the Euler characteristic of a
polyhedral complex in dimension~3.

We consider now the game where
$P$ is the dual cyclic polytope $C_4(4+n)^\Delta$, and
$Q$ is the dual cyclic polytope $C_n(4+n)^\Delta$ labeled
with the swap permutation in (\ref{pi}).
Let $n$ be even, $n=2k$.
According to Theorem~\ref{t-cyc} and (\ref{sigma}),
the resulting equilibrium number is
$\sigma(2,k)+\sigma(1,k-1)$, that is,
\begin{equation}
\label{4n}
\arraycolsep 0.1em
\begin{array}{cl}
&\binom{2}{0}\binom{k+2}{2}+
\binom{2}{1}\binom{k+1}{2}+
\binom{2}{2}\binom{k}{2}+
\binom{1}{0}\binom{k}{1}+
\binom{1}{1}\binom{k-1}{1} 
\\[1ex]
=&2k^2+2k+1~+~2k-1
~=~2k^2+4k
~=~{n^2}/2+2n~.
\end{array}
\end{equation}
The difference to the upper bound $(n^2+5n+4)/2$ is
$n/2+2$ or $k+2$\,.
This agrees with the observation that the only Gale evenness
strings for $C_4(4+n)^\Delta$ that are not swappable are
those given by a substring 011110 followed by $n-2$ 0's,
or any cyclic shift by an even amount of such a string;
there are $k+2$ such even cyclic shifts. 
We now use those substrings 011110 on an even position to
create suitable obstructions that lose one equilibrium
vertex each time, using the triangle bound in
Lemma~\ref{l-odd}.

\begin{theorem}
\label{4opt}
For a $4\times n$ game with even $n$ where $P$ and $Q$ in
$(\ref{PQ})$ are dual cyclic polytopes, the swap permutation
$(\ref{pi})$ for the labels of $Q$ gives the maximum number
of equilibria. 
\end{theorem}

\proof
Let $n=2k$.
By assumption, $P=C_4(4+n)^\Delta$.
We claim that $P$ has $k+2$ disjoint triangles as faces,
each of which loses one equilibrium vertex by
Lemma~\ref{l-odd}.
As shown in (\ref{4n}), this is exactly the number of
non-equilibrium vertices when using the swap permutation
(\ref{pi}).
To show the claim, consider the following triple of Gale
evenness strings:
\begin{equation}
\label{ge4}
\begin{array}{l}
01\,11\,10~0^{2k-2}\\
11\,01\,10~0^{2k-2}\\
11\,11\,00~0^{2k-2}\\
\end{array}
\end{equation}
This triple of Gale evenness strings represents a triangle
as a face of $P$, because any two of them have three facets
(with bits~1) in common; the intersection of these three
facets defines an edge of~$P$.
The first string in (\ref{ge4}) is not swappable, but the
other two are.
Each cyclic shift of the strings in (\ref{ge4}) by an even
number of places creates a new, disjoint triangle.
There are $k+2$ such shifts, which proves the claim.
\endproof

Dual cyclic polytopes are dual-neighborly polytopes, which
have the maximum possible number of vertices.
They have a very regular structure described by the Gale
evenness condition, which allows to identify matching vertex
pairs via the swap permutation (\ref{pi}).
One may therefore conjecture that Theorem~\ref{4opt}
describes games with the maximum possible number of
equilibria.
However, this conjecture is near-impossible to prove when
considering all dual-neighborly polytopes because they
are extremely numerous, even for known constructions
\citep{shemer1982}, and not fully characterized.
There are only 37 combinatorial types of simple 4-polytopes
with 8 facets, of which 3 are dual-neighborly.
There are already 159375 dual-neighborly 5-polytopes with 10
facets \citep{firsching-2017}, among an unknown number of
simple polytopes.

We will give an asymptotic bound for general $d\times n$
games in Subsection~\ref{s-drows}.

\subsection{$6\times n$ games}
\label{s-6}

The considerations of Theorem~\ref{4opt} can be extended to
games with six rows.

\begin{theorem}
\label{6opt}
For a $6\times n$ game with even $n$ where $P$ and $Q$ in
$(\ref{PQ})$ are dual cyclic polytopes, the swap permutation
$(\ref{pi})$ for the labels of $Q$ gives the maximum number
of equilibria. 
\end{theorem}

\proof
Let $n=2k$.
By assumption, $P=C_6(6+n)^\Delta$.
By (\ref{GE}), $P$ has $\binom{3+2k}{3}+ \binom{2+2k}{2}$
many vertices, represented by the Gale evenness strings.
With the labels $\pi(i)$ of $Q$ as in (\ref{pi}), the
number of equilibrium vertices of~$P$ (given by the swappable
bitstrings) is $\sigma(3,k)+\sigma(2,k-1)$ according to
Theorem~\ref{t-cyc}.
Because, using (\ref{sigma}), 
\[
\arraycolsep.1em
\begin{array}{rclrcl}
\binom{3+2k}{3}+\binom{2+2k}{2}
&=&
(4k^3+18k^2+20k+6)/3\,,
\\[1ex]
\sigma(3,k)&=&(4k^3+6k^2+8k+3)/3,
\qquad
\sigma(2,k-1)=2k^2-2k+1\,,
\end{array}
\]
the number of
non-equilibrium vertices is
\begin{equation}
\label{ub6}
\binom{3+2k}{3}+
\binom{2+2k}{2}
-\sigma(3,k)-\sigma(2,k-1)=2k(k+3)\,.
\end{equation}
Using the simplex bound of Lemma~\ref{l-simplex},
we now show that there are exactly $2k(k+3)$ obstructions
(non-equilibrium vertices) in~$P$, which shows that the labeling of
$Q$ via $\pi$ is optimal.

First, we consider tetrahedra (simplices with four vertices) that arise from 
the non-swappable Gale evenness strings with a substring
01111110 starting on an \textit{even position} (that is,
preceded by an even number of bits), for example at the
beginning of the string.
The respective tetrahedron is given by the four vertices
\begin{equation}
\label{ge6}
\begin{array}{l}
0\underline 1\,1\underline 1\,1\underline 1\,10~0^{2k-2}\\
1\underline 1\,0\underline 1\,1\underline 1\,10~0^{2k-2}\\
1\underline 1\,1\underline 1\,0\underline 1\,10~0^{2k-2}\\
1\underline 1\,1\underline 1\,1\underline 1\,00~0^{2k-2}\\
\end{array}
\end{equation}
where the last two bitstrings are swappable but the first
two are not.
The underlined bits 1 in second, fourth, and sixth position
indicate the common facets of these vertices.
The intersection of these three facets is the tetrahedron, a
face of~$P$.
By Lemma~\ref{l-simplex}, at least two vertices of this
tetrahedron cannot be equilibrium vertices, no matter how
the polytopes are labeled.

Any simultaneous cyclic shift of the vertices in (\ref{ge6})
by an even amount creates a new tetrahedron.
There are $k+3$ such even cyclic shifts
(half of the length $2k+6$ of the entire bitstring). 
All vertices of the shifted tetrahedron are also new:
This clearly holds separately for each row of (\ref{ge6}),
and the first and last rows always differ because their
substring 111111 starts on an odd and even position,
respectively.
Hence, these tetrahedra create $2(k+3)$ obstructions.

In addition, we consider triangles that arise from
substrings 011110 starting on an even position, followed by a
substring 11 arbitrarily thereafter (using a circular
wraparound of the bitstring) but with at least one 0 before
meeting again 011110.
These triangles, when started at the beginning of the
string, are described by
\begin{equation}
\label{gej}
\begin{array}{l}
01\,11\,10~0^j11\,0^{2k-3-j}\underline0\\
11\,01\,10~0^j11\,0^{2k-3-j}\underline0\\
11\,11\,00~0^j11\,0^{2k-3-j}\underline0\\
\end{array}
\qquad\hbox{for $0\le j\le 2k-3$\,.}
\end{equation}
There are $2k-2$ possible choices for~$j$, where each of
them defines a different triangle.
These triangles are mutually disjoint:
The first and last row in (\ref{gej}) differ because their
substring 1111 starts on an odd and even position,
respectively; this holds also after any even cyclic shift.
They also differ from the second row in (\ref{gej}), which
has no substring 1111.
The second row, after any even cyclic shift, also contains
for $j=0$ and for $j=2k-3$ the substring 11011011,
but the former on an even position and the latter on an odd
position, so they are also different.

Furthermore, the tetrahedra in (\ref{ge6}) and the triangles
in (\ref{gej}) have no common vertices.
For those that have the substring 1111 on an odd position,
it is not possible to create the second row in
the tetrahedron (\ref{ge6}) via a cyclic shift of the first
row of (\ref{gej}) for $j=2k-3$ because the underlined 0
in (\ref{gej}) differs from the first underlined bit~1 in
(\ref{ge6}).
Similarly, for the vertices that have the substring 1111 on
an even position, the third rows of (\ref{ge6}) and
(\ref{gej}) always differ for any even cyclic shift.

Consequently, all the vertices in (\ref{ge6}) and
(\ref{gej}) for any even cyclic shift are different.
By Lemma~\ref{l-simplex}, each triangle in (\ref{gej}) has
at least one non-equilibrium vertex.
These are $2k-2$ triangles for the choices of~$j$,
multiplied by $k+3$ even cyclic shifts, which
give $(2k-2)(k+3)$ obstructions.
Together with the $2(k+3)$ obstructions from the tetrahedra
(\ref{ge6}), these are $2k(k+3)$ obstructions in total,
exactly the difference in (\ref{ub6}), which proves the
claimed optimality.
\endproof

Theorem~\ref{6opt} can be applied to $n=6$, where it shows
that the swap permutation $\pi$ in (\ref{pi}) is an optimal
labeling of~$Q$.
It gives rise to 76 equilibria in a $6\times 6$ game; each
polytope has 112 vertices.
The permutation $\pi$ was stated in \cite{vS1999}.
Theorem~\ref{6opt} shows that it is optimal.
This can also be verified by testing all 12! labelings
of~$Q$, but is here proved to be optimal for any even~$n$.

\subsection{Asymptotically almost all vertices are equilibria}
\label{s-drows}

In this section we consider $d\times n$ games where $P$ and
$Q$ in (\ref{PQ}) are dual cyclic polytopes,
$P=C_d(d+n)^\Delta$ and $Q=C_n(d+n)^\Delta$
(after a suitable affine transformation, which does not
affect their combinatorial structure).
We assume $d\le n$, so that $P$ is the lower-dimensional
polytope with the smaller number of vertices $g(d,n)$, which
is therefore an upper bound on the number of equilibria.
The vertices of $P$ and $Q$ are encoded by the Gale evenness
strings, where the 1's indicate the facets that the vertex
lies on.
The facets of $P$ are labeled in increasing order of the
positions in the bitstring, and the labels of~$Q$ are
permuted by the swap permutation $\pi$ in~(\ref{pi}); if
$d+n$ is odd, then the last label $d+n$ is kept fixed as
in~(\ref{swap9}).
An equilibrium of $P\times Q$ corresponds to a pair of
bitstrings that both fulfill Gale evenness and that are
complementary when this labeling of~$Q$ is taken into
account.

We write $d=2\ell$ or $d=2\ell+1$, and $n=2k$ or $n=2k+1$,
so that $\ell=\lfloor d/2\rfloor$ and $k=\lfloor n/2\rfloor$.
For fixed~$d$ the number of vertices $g(d,n)$ grows in $n$ like
$n^{\ell}/\ell!$ for even~$d$, and like $2\,n^{\ell}/\ell!$
for odd~$d$.
In both cases it is of order $\Theta(n^{\lfloor
d/2\rfloor})$, which is the asymptotic version of the Upper
Bound Theorem~\ref{UBT} \citep{Mul1994}.

That the labeling of~$Q$ matters is demonstrated by the
following lemma, which considers the identity permutation.

\begin{lemma}
If $P$ and $Q$ are identically labeled, where $d=2\ell$ and
$n=2k$, then the fraction of equilibrium vertices among all
vertices of $P$ tends to $1/2^\ell$ as $k$ grows.
\end{lemma}

\proof
Because $P$ and $Q$ are identically labeled, an equilibrium
vertex of $P$ is a Gale evenness string if and only if both its
1's and 0's come in pairs.
This means that the string is composed of
$\ell$ substrings 11 and $k$ substrings 00, with $\binom{\ell+k}{\ell}$
possibilities,
or it has a single 1 at each end with 
$\ell-1$ substrings 11 and $k$ substrings 00 in between,
with $\binom{\ell-1+k}{\ell-1}$ possibilities.
The total number of these strings is therefore
\begin{equation}
\label{both}
\binom{\ell+k}{\ell} +\binom{\ell-1+k}{\ell-1}
=
\binom{\ell+k}{\ell} \left(1+\frac\ell{\ell+k}\right)
\,.
\end{equation}
The fraction $F$ of these strings compared to all $g(2\ell,2k)$
Gale evenness strings is given by
\begin{equation}
\label{Rn}
F = 
\frac
{\binom{\ell+k}{\ell} \left(1+\frac\ell{\ell+k}\right)}
{\binom{\ell+2k}{\ell} \left(1+\frac\ell{\ell+2k}\right)}
= 
\frac
{\left(1+\frac\ell{\ell+k}\right)}
{\left(1+\frac\ell{\ell+2k}\right)}
\prod_{i=1}^{\ell}\frac{k+i}{2k+i}
= 
\frac
{\left(1+\frac\ell{\ell+k}\right)}
{\left(1+\frac\ell{\ell+2k}\right)}
\prod_{i=1}^{\ell}\left(1-\frac{1}{2+i/k}\right)\,.
\end{equation}
Because $1-\frac{1}{2+i/k}$ is increasing in~$i$, we have
\begin{equation}
\label{bound}
\frac1{2^\ell}=\left(1-\frac{1}{2}\right)^\ell
<
\left(1-\frac{1}{2+1/k}\right)^\ell
<
F
\le
\frac
{\left(1+\frac\ell{\ell+k}\right)}
{\left(1+\frac\ell{\ell+2k}\right)}
\left(1-\frac{1}{2+\ell/k}\right)^\ell
\end{equation}
where the right term clearly tends to $1/2^\ell$ for large~$k$.
\endproof

We now use the swap permutation $\pi$ for the labels of~$Q$.
The following theorem generalizes Theorem~\ref{t-cyc}, which
gave the equilibrium number for even $d$ and~$n$, to all
combinations of parities of $d$ and~$n$, and shows in each
case that these equilibria are approximately \emph{all} of
the vertices of~$P$, up to a vanishing fraction of order
$O(1/n)$.

\begin{theorem}
\label{t-limit}
Let $P=C_d(d+n)^\Delta$ and $Q=C_n(d+n)^\Delta$ with $Q$
labeled by the swap permutation~$\pi$, and let $d\le n$.
Then the number of equilibrium vertices of~$P$ is
\begin{equation}
\label{allcounts}
\begin{array}{ll}
\sigma(\ell,k)+\sigma(\ell-1,k-1) &
\quad\hbox{if $d=2\ell$ and $n=2k$,}\\[.5ex]
\sigma(\ell,k)+\sigma(\ell-1,k)   &
\quad\hbox{if $d=2\ell$ and $n=2k+1$,}\\[.5ex]
\sigma(\ell,k)+\sigma(\ell,k-1)   &
\quad\hbox{if $d=2\ell+1$ and $n=2k$,}\\[.5ex]
2\,\sigma(\ell,k)                 &
\quad\hbox{if $d=2\ell+1$ and $n=2k+1$.}
\end{array}
\end{equation}
In each case, the fraction of equilibrium vertices among all
vertices of~$P$ tends to~1 as $n$ grows, and the fraction of
non-equilibrium vertices is of order $O(1/n)$.
\end{theorem}

\proof
We first describe the equilibrium vertices in each case.

The case $d=2\ell$ and $n=2k$ has been stated in
Theorem~\ref{t-cyc} .
Recall that under the labeling $\pi$ every
$(\ell,k)$-swappable string, composed of substrings 0110,
11, and 00, is an equilibrium vertex, because its complement
is composed of 1001, 11, and 00, which under $\pi$ becomes
$(k,\ell)$-swappable and hence fulfills Gale evenness.
There are $\sigma(\ell,k)$ such strings, see~(\ref{sigma}).
In addition the $(\ell-1,k-1)$-swappable strings preceded by
10 and succeeded by 01 are equilibrium vertices, which gives
the term $\sigma(\ell-1,k-1)$.

If $d=2\ell$ and $n=2k+1$, then $d+n$ is odd and the last
label of~$Q$ is kept fixed.
A Gale evenness string with $2\ell$ 1's and $2k+1$
0's is an equilibrium vertex if and only if it is
$(\ell,k)$-swappable followed by a terminal~0, or an initial
substring 10 followed by an $(\ell-1,k)$-swappable string
followed by a terminal~1, as already observed in the special
case $d=4$, $n=5$ before~(\ref{ge49}).
The two families have $\sigma(\ell,k)$ and $\sigma(\ell-1,k)$
members, respectively.
Note that we do not need that $d\le n$.

If $d=2\ell+1$ and $n=2k$, then this game is obtained from
the even--odd case by exchanging $d$ and~$n$.
We exchange the two polytopes $P$ and~$Q$, and apply the
swap permutation $\pi$ to the labels of both polytopes;
because $\pi$ is its own inverse, the first polytope
(formerly~$Q$) is then as assumed labeled with the identity
permutation and the second polytope is labeled with~$\pi$.
Applying the even--odd count $\sigma(\ell,k)+\sigma(\ell-1,k)$
with the roles of $d$ and $n$ interchanged gives the
equilibrium number $\sigma(\ell,k)+\sigma(\ell,k-1)$.

If both $d$ and $n$ are odd, $d=2\ell+1$ and $n=2k+1$, then
$d+n$ is even and the swap permutation $\pi$ of~(\ref{pi})
applies without a fixed last label.
Here every Gale evenness string of~$P$ has an isolated~1 at
one of its two ends, because it must begin or end with an
odd number of 1's, and a run of three or more 1's would
violate complementarity under~$\pi$.
The equilibrium vertices are therefore the $(\ell,k)$-swappable
strings preceded by 10 or succeeded by~01, which gives
$2\,\sigma(\ell,k)$ strings.

For the asymptotics it suffices to keep only the leading
term $\sigma(\ell,k)$; the second term in~(\ref{allcounts})
is nonnegative and of strictly lower order, and dropping it
only weakens the bound.
We bound $\sigma(\ell,k)$ from below by replacing each
binomial coefficient $\binom{k+\ell-i}{\ell}$ in~(\ref{sigma})
by its smallest value $\binom{k}{\ell}$, attained at $i=\ell$:
\begin{equation}
\label{multi}
\sigma(\ell,k)
=\sum_{i=0}^\ell\binom \ell i \binom{k+\ell-i}{\ell}
\ge\binom{k}{\ell}\sum_{i=0}^\ell\binom \ell i
=\binom k \ell\,2^\ell\,.
\end{equation}
Let $E$ denote the number of equilibrium vertices
in~(\ref{allcounts}), which has leading term
$\sigma(\ell,k)$ in all cases.
For even~$d$ we have $g(d,n)=\binom{\ell+n}{\ell}
(1+\frac{\ell}{\ell+n})$ by~(\ref{GE}), and for odd~$d$ we have
$g(d,n)=2\binom{\ell+n}{\ell}$ by~(\ref{GEodd}).
If both $d$ and $n$ are odd, then $E$ and $g(d,n)$ have a
factor~2 that cancels, so that the following inequalities
apply in all cases:
\begin{equation}
\label{SG}
\arraycolsep.2em
\begin{array}{rcl}
\dfrac{E}{g(d,n)}
&\ge&
\dfrac{\sigma(\ell,k)}{\,\binom{\ell+n}{\ell}\,(1+O(1/n))\,}
\;\ge\;
\dfrac{\binom{k}{\ell}\,2^\ell}{\binom{\ell+2k}{\ell}}\,(1-O(1/n))
\\[2.5ex]
&=&
\displaystyle
\prod_{i=1}^{\ell}\frac{2(k-\ell+i)}{2k+i}\,(1-O(1/n))
=
\prod_{i=1}^{\ell}\Bigl(1-\frac{2\ell-i}{2k+i}\Bigr)\,(1-O(1/n)).
\end{array}
\end{equation}
Since $\prod_{i=1}^{\ell}\bigl(1-\frac{2\ell-i}{2k+i}\bigr)
\ge\bigl(1-\frac{2\ell-1}{2k+1}\bigr)^{\ell}=1-O(1/k)$,
it follows that $E/g(d,n)$ tends to~1, and the fraction of
non-equilibrium vertices is of order $O(1/k)=O(1/n)$.
\endproof

The bound~(\ref{SG}) only shows that the non-equilibrium
vertices are a vanishing fraction of order $O(1/n)$.
The next theorem states their number precisely, and in
particular shows that this fraction is of order exactly
$1/n$, and not smaller.

\begin{theorem}
\label{t-deficit}
Let $\ell\ge2$.
Under the assumptions of Theorem~\ref{t-limit}, the number of
non-equilibrium vertices of~$P$ equals
\begin{equation}
\label{deficits}
\frac{n^{\ell-1}}{\ell!}
\left\{
\begin{array}{lll}
\binom{\ell}{2}   & \hbox{if $d=2\ell$,}& n=2k \\[.5ex]
\binom{\ell+1}{2} & \hbox{if $d=2\ell$,}& n=2k+1 \\[.5ex]
\ell(\ell+1)      & \hbox{if $d=2\ell+1$,}& n=2k  \\[.5ex]
\ell(\ell+1) & \hbox{if $d=2\ell+1$,}& n=2k+1 \end{array}
\right\}
\;+\;O(n^{\ell-2}).
\end{equation}
Equivalently, the fraction of non-equilibrium vertices of
$P$ among
all $g(d,n)$ vertices is $\binom{\ell}{2}/n+O(1/n^2)$ in the
even--even case, and $\binom{\ell+1}{2}/n+O(1/n^2)$ in the
three other cases.
\end{theorem}

\proof
Since $\ell$ is fixed, all quantities are polynomials
in~$k$ of degree at most~$\ell$, and we compare their two
leading terms; the two leading coefficients in~$k$ determine
the two leading coefficients in~$n=2k$ or $n=2k+1$, since
$n=2k+O(1)$.
Expanding the falling factorial,
\[
\arraycolsep.2em
\renewcommand{\arraystretch}{1.3}
\begin{array}{rcl}
\displaystyle
\binom{k+\ell-i}{\ell}
&=&
\displaystyle
\frac{1}{\ell!}\prod_{j=0}^{\ell-1}(k+\ell-i-j)
\;=\;
\frac{k^{\ell}}{\ell!}
+\frac{k^{\ell-1}}{\ell!}\sum_{j=0}^{\ell-1}(\ell-i-j)
+O(k^{\ell-2})
\\
&=&
\displaystyle
\frac{k^{\ell}}{\ell!}
+\frac{k^{\ell-1}}{\ell!}\Bigl(\ell(\ell-i)-\tbinom{\ell}{2}\Bigr)
+O(k^{\ell-2}).
\end{array}
\]
Summing against $\binom{\ell}{i}$ and using
$\sum_{i}\binom{\ell}{i}=2^{\ell}$ and
$\sum_{i}\binom{\ell}{i}i=\ell\,2^{\ell-1}$, the coefficient
of $k^{\ell-1}/\ell!$ in $\sigma(\ell,k)$ is
$\sum_{i}\binom{\ell}{i}\bigl(\ell(\ell-i)-\binom{\ell}{2}\bigr)
=2^{\ell}(\ell^{2}-\tfrac{\ell^{2}}{2}-\binom{\ell}{2})
=2^{\ell}\cdot\tfrac{\ell}{2}$, so that
\begin{equation}
\label{sigmaexp}
\sigma(\ell,k)
=\frac{2^{\ell}k^{\ell}}{\ell!}
+\frac{2^{\ell}k^{\ell-1}}{\ell!}\cdot\frac{\ell}{2}
+O(k^{\ell-2}).
\end{equation}
The lower-order terms $\sigma(\ell-1,k-1)$, of
$\sigma(\ell-1,k)$, or of $\sigma(\ell,k-1)$ are computed the
same way; each is of degree $\ell-1$ in~$k$ and contributes
to~(\ref{deficits}) only through its own leading term
$2^{\ell-1}k^{\ell-1}/(\ell-1)!$ in the first two cases,
respectively $2^{\ell}k^{\ell-1}/\ell!$ for
$\sigma(\ell,k-1)$.
When subtracting the equilibrium number $E$ in~(\ref{allcounts})
from the vertex number $g(d,n)$ in~(\ref{GE}) or~(\ref{GEodd}),
the $k^{\ell}$-terms cancel, and collecting the coefficients
of $k^{\ell-1}$ gives the leading terms stated
in~(\ref{deficits}).

For instance, the deficit equals $k+2$ for $\ell=2$ and
$2k(k+3)$ for $\ell=3$ in the even--even case, which matches the
exact counts in~(\ref{4n}) and~(\ref{ub6}); it equals
$3(k+1)$ for the $4\times(2k{+}1)$ game, which for $k=2$ gives
the~$9$ non-equilibrium vertices of the $4\times5$ game with
$27-9=18$ equilibria of Lemma~\ref{l-17}.
The fraction statements follow by dividing~(\ref{deficits})
by $g(d,n)$, whose leading term is $n^{\ell}/\ell!$ for
even~$d$ and $2\,n^{\ell}/\ell!$ for odd~$d$; the factor~2 in
the odd case turns the deficit coefficient $\ell(\ell+1)$ into
the fraction constant $\ell(\ell+1)/2=\binom{\ell+1}{2}$.
\endproof

For a $4\times7$ game, for example, Theorem~\ref{t-limit}
gives $\sigma(2,3)+\sigma(1,3)=25+7=32$ equilibria, and
Theorem~\ref{t-deficit} an $O(1/n)$ non-equilibrium fraction
with leading deficit $3(k+1)=3\cdot4=12$; indeed the vertex
number is $g(4,7)=44$, and $44-12=32$.

Remark: For odd $d=2\ell+1$ and odd $n=2k+1$, the
$2\sigma(\ell,k)$ equilibria of Theorem~\ref{t-limit} form
all but an $O(1/n)$ fraction of the vertices.
However, this construction is not optimal.
As already observed by \citet{vS1999},
it is improved by \emph{doubling} the equilibrium number of
the $(d-1)\times(n-1)$ game $(A,B)$ obtained for even $d-1$
and $n-1$:
adding a unit vector as an additional row and
column yields the $d\times n$ game with payoff matrices
$\binom{A~~\0}{\0\T~1}$ and $\binom{B~~\0}{\0\T~1}$,
each of whose equilibria gives rise to two equilibria of the
augmented game (one where the new strategy is unplayed, and
one where it is played with positive probability, using that
equilibrium payoffs are positive).
By Theorem~\ref{t-cyc}, the resulting game has the larger
equilibrium number
$2\bigl(\sigma(\ell,k)+\sigma(\ell-1,k-1)\bigr)$ compared
to~$2\sigma(\ell,k)$ in (\ref{allcounts}),
for example $32$ for $d=n=5$ compared to~$26$.
The non-equilibrium coefficient in (\ref{deficits}) then 
drops from $\ell(\ell+1)$ to $\ell(\ell-1)$,
and the corresponding fraction to the even--even term
$\binom{\ell}{2}/n$.
\citet{ITvS25} recently proved that $32$ is in fact the
maximum number of equilibria for a $5\times5$ game.

\subsection{Cyclic polytopes: computational evidence} 
\label{s-compute}

We continue to consider $P$ and $Q$ given as dual cyclic
polytopes in dimension $d$ and~$n$, respectively, for even
$d$ and~$n$, now for $d\ge8$.
The conjecture is that $Q$ labeled with the swap
permutation $\pi$ in (\ref{pi}) maximizes the number of
equilibria for this polytope pair.

The proofs of tight bounds in Theorems \ref{4opt} and \ref{6opt}
unfortunately do not extend to $d=8$ and beyond.
While it is possible to create obstruction simplices
similar to (\ref{ge6}) or (\ref{gej}), they are not
necessarily disjoint or not numerous enough.
For example, a non-swappable bitstring may contain
twice the substring 011110, in which case the respective
triangles similar to (\ref{ge4}) are interdependent and not
disjoint.
With growing $d$, the Gale evenness strings will contain more
and more repeated substrings like 011110 or 01111110 in
addition to the swappable substring 0110.
For that reason, the simplex bound becomes less and less
tight. 

However, we have computational evidence that the swap
permutation $\pi$ gives the optimal labeling of~$Q$.
We have checked all permutations as labelings for the case $d=n=8$. 
We have reduced the number 16! of all possible permutations 
to 15! by considering only permutations $\rho$ of $[16]$
that keep element 16 fixed, $\rho(16)=16$ (the
resulting computation took 9 days on a standard laptop).

The swap permutation $\pi$ is also optimal in all small
cases with odd~$d$ or~$n$ that we have checked exhaustively:
for even $d$ and odd~$n$ these are the $4\times5$ and
$4\times7$ games (Lemma~\ref{l-17} and the remark following
Theorem~\ref{t-limit}), and for odd~$d$ and even~$n$ the
$3\times4$, $3\times6$, and $5\times4$ games; the maximum
equilibrium numbers agree with the counts of
Theorem~\ref{t-limit}.

More generally, any permutation $\rho$ of the labels in
$[d+n]$ can be changed such that $\rho(d+n)=d+n$ 
with the same equilibrium numbers.
This uses the cyclic symmetry of the Gale evenness strings,
and that these strings can also be reversed (written
backwards).
With this normalization, $\pi$ in (\ref{pi}) appears for
$d=n=4,6,8$ as the \textit{unique} (up to symmetry)
equilibrium-maximizing permutation $\rho$ in one of the
following four equivalent forms, illustrated for $d=4$, with
12345678 mapped to one of
\begin{equation}
\label{maxfix}
7\underline21\underline43\underline65\underline8,\qquad
\underline167\underline4\underline523\underline8,\qquad
3\underline25\underline47\underline61\underline8,\qquad
56\underline3\underline412\underline7\underline8
\end{equation}
where we have underlined the fixed points.
The first of these permutations is~$\pi$, originally given by
$12345678\mapsto 21436587$ as in (\ref{swap}), followed by a
cyclic shift by~1. The second permutation in (\ref{maxfix})
is $\hat \pi: 12345678\mapsto 83254761$ with the entire
string of labels written in reverse
($\hat\pi$ is a cyclic shift of $-1$ followed by $\pi$
followed by a cyclic shift of~$1$).
The third permutation in (\ref{maxfix}) is $\hat\pi$
followed by a cyclic shift by~$-1$. 
The fourth permutation in (\ref{maxfix}) is $\pi$, reversed,
followed by a cyclic shift by~$-2$. 

Considering all permutations for small even $d$ and $n$
demonstrates a significant gap of equilibrium numbers for 
the permutation, up to symmetry, that (according to the
computational evidence for $d=n=8$) creates the
\textit{second-largest} number of equilibria.
This permutation $\mu$ of $[d+n]$ is for $d=n=4$ given by
$12345678\mapsto 4321\,6587$ and generally by
\begin{equation}
\label{mu}
\mu(i)=\begin{cases}
5-i & \hbox{if $i\in\{1,2,3,4\}$,}\\
i+1 & \hbox{if $i\in\{ 5,\ldots,d+n\}$ and $i$ is odd,}\\
i-1 & \hbox{if $i\in\{ 6,\ldots,d+n\}$ and $i$ is even.}\\
\end{cases}
\end{equation}
The permutation $\mu$ reverses the first four bits of 
a bitstring and is otherwise identical to the swap
permutation $\pi$. 
Interestingly, $\mu$ \textit{gains}, compared to $\pi$, Gale
evenness strings in $P$ that start with $011110$ followed by
a swappable bitstring of the right number of 0's and 1's.
However, it \textit{loses} the swappable bitstrings that
start with $0110$, and related strings that are equilibria
under~$\pi$. 
The equilibrium numbers for $\mu$ are 12 for $d=4$,
then 60 for $d=6$, and 308 for $d=8$,
compared to 16, 76, and 384 for $\pi$.
The following proposition states them in general.

\begin{proposition}
\label{p-mu}
Consider a $d\times n$ game with even $d=2\ell$ and $n=2k$
where $P$ and $Q$ in $(\ref{PQ})$ are dual cyclic polytopes,
and the permutation $\mu$ in $(\ref{mu})$ for the labels of~$Q$.
Then the number of equilibria is
$\sigma(\ell,k)-\sigma(\ell-2,k-2)$,
whereas for $\pi$ it is $E(d,n)= \sigma(\ell,k)+\sigma(\ell-1,k-1)$.
The difference is exactly $E(d-2,n-2)$.
\end{proposition}

\proof
We first describe the equilibrium strings in $P$ that work
for $\mu$ but not for $\pi$.
They are exactly given by strings of the form $0111\,10\,S$
or $1110\,S\,01$ for a bitstring $S$ that is
$(\ell-2,k-1)$-swappable.
Observe that $S$ has the right number of $d-4$ bits~1 and
$n-2$ bits~0.
The first string $0111\,10\,S$ maps under $\mu$ to
$1110\,01\,\mu(S)$ which when complemented becomes
$0001\,10\,S'$ for a swappable substring $S'$
and fulfills Gale evenness.
However, under $\pi$ it maps to
$1011\,01\,\pi(S)$ that when complemented fails Gale
evenness.
The second string $1110\,S\,01$ maps under $\mu$ to
$0111\,\mu(S)\,10$ which when complemented becomes
$1000\,S'\,01$ which fulfills Gale evenness, but
similarly fails for~$\pi$.

Second, the equilibrium strings in $P$ that work for $\pi$
but fail for $\mu$ are of the following four kinds:

(a)~~~
$1101\,10\,S$ where $S$ is $(\ell-2,k-1)$-swappable.

(b)~~~
$1011\,S\,01$ where $S$ is $(\ell-2,k-1)$-swappable.

(c)~~~
$0110\,T$ where $T$ is $(\ell-1,k-1)$-swappable.

(d)~~~
$1001\,10\,U\,01$ where $U$ is $(\ell-2,k-2)$-swappable.

\noindent
All of these strings give equilibria under $\pi$.
Consider how they map under $\mu$ and then under
complementation do not fulfill Gale evenness; 
we denote by $\mu(S)$ the restriction of $\mu$ to the
positions occupied by~$S$:

(a)~~~
$1011\,01\,\mu(S)$, complements to $0100\,10\,S'$

(b)~~~
$1101\,\mu(S)\,10$, complements to $0010\,S'\,01$

(c)~~~
$0110\,\mu(T)$, complements to $1001\,T'$

(d)~~~
$1001\,01\,\mu(U)\,10$, complements to $0110\,10\,U'\,01$

\noindent
where $S',T',U'$ are swappable strings.

The above are exactly the Gale evenness strings where $\pi$
and $\mu$ differ in their equilibrium properties.
This holds because $\pi$ and $\mu$ differ only on the first
four bit positions:
Initial substrings 0100, 0010, 0101, and 1010 fail Gale
evenness.
Initial substrings 0000, 0011, 1100, and 1111 followed by a
swappable string work for both $\pi$ and $\mu$.
Strings $1000\,S\,01$ for a swappable string $S$ also work
for both $\pi$ and $\mu$, 
and the same holds for $0001\,10\,S$,
but strings $0001(11)^t10\,S$ for $t\ge1$ fail for both
$\pi$ and~$\mu$.
This covers all 16 cases of initial four-bit substrings.

In terms of equilibrium counts, the
$2\cdot \sigma(\ell-2,k-1)$ gains with substrings $S$ under
$\mu$ are compensated by the losses in (a) and (b) under~$\pi$.
The number of losses in (c) and (d) where $\pi$ gives an
equilibrium but $\mu$ does not is
$\sigma(\ell-1,k-1)+\sigma(\ell-2,k-2)=E(d-2,n-2)$ as
claimed.
The number of equilibria under $\mu$ is the number
$\sigma(\ell,k)+\sigma(\ell-1,k-1)$ under $\pi$ minus
$\sigma(\ell-1,k-1)+\sigma(\ell-2,k-2)$. 
\endproof

We state the following conjecture.

\begin{conjecture}
\label{c-cyc}
Consider $d\times n$ games for any even $d$ and $n$ and dual
cyclic polytopes $P$ and~$Q$.
Then the swap permutation $\pi$ in $(\ref{pi})$ is, up to
the cyclic and reversal symmetry of the Gale evenness
strings, the unique labeling of $Q$ with the maximum number
of equilibria.
 \end{conjecture}

As extra information and guidance, we may add a secondary
conjecture:
The permutation $\mu$ in $(\ref{mu})$ is the, for $n\ge
d\ge6$ up to symmetry unique, permutation with the
second-largest number of equilibria. 
(For $d=n=4$ the identity permutation gives also the
second-largest number 12 of equilibria, so $\mu$ is not
unique.)

In order to prove Conjecture~\ref{c-cyc} for $d=8$, the
clique bound used in the proofs of Theorems
\ref{4opt} and~\ref{6opt} is smaller than the observed
``gap'' $E(6,n-2)$ in equilibrium counts in
Proposition~\ref{p-mu} between the permutations $\pi$ and~$\mu$.
For $n=8$ that gap is $384-308=76$, and $\mu$ is
computationally verified as the second-largest permutation.
However, these gaps go in opposite direction: the clique
bound is 408, obtained via 252 obstructions, and thus larger by 24
than the optimum 384.
Hence, it is not clear if the two proof approaches can be
usefully combined.

An even stronger conjecture is the following.

\begin{conjecture}
\label{c-even}
Consider nondegenerate $d\times n$ games for any even $d$ and $n$.
Then the maximum number of equilibria is obtained
for dual cyclic polytopes $P$ and $Q$ with $Q$ labeled
with the swap permutation $\pi$ in~$(\ref{pi})$.
\end{conjecture}

The dual cyclic polytopes maximize the vertex counts by
the Upper Bound Theorem, and give rise to descriptions with 
a high symmetry that allow the construction of~$\pi$.
Any further dual-neighborly polytopes will presumably be
less symmetric.
However, they are by far too numerous to collect any
computational evidence to support Conjecture~\ref{c-even}.

\subsection{Discussion}
\label{s-discuss}

We conclude with a health warning.
Dual cyclic polytopes play a central role in our 
constructions of bimatrix games with a large number of
equilibria.
They have also been used to construct hard instances for
equilibrium-finding algorithms
\citep{SvS2004,SvS2006,SvS2016}.

Their defining inequalities in (\ref{cyc}), which are
based on the moment curve, lead to matrix entries of large
varying sizes that nevertheless need to be very accurate
because otherwise the equilibrium numbers are not correct,
nor does a Lemke-Howson path become as long as predicted
in~\cite{SvS2006}.
An alternative is to use the trigonometric moment curve,
if ${d=2\ell}$,
that maps $t$ to $(\cos t, \sin t, \cos 2t, \sin 2t, \ldots,\cos
\ell t, \sin \ell t)$ 
\citep{caratheodory1911,gale1963}.
Its values need to be rounded to rational numbers in order
to get suitable matrix entries.
Evidence for small values of $d$ indicates that the
necessary number of \textit{digits} for these rational
numbers needs to be proportional to $d$ in order to
correctly obtain the claimed combinatorial structure of Gale
evenness strings, similarly to numbers obtained from the
standard moment curve.

This suggests that cyclic polytopes are ``numerically
brittle'' and therefore not ideally suited to provide hard
computational instances for equilibrium-finding algorithms
in practice.

\needspace{3cm}
\addcontentsline{toc}{section}{References}
\small
\bibsep.8ex plus.1ex minus.05ex
\bibliography{bib-rect.bib}
\bibliographystyle{book}

\end{document}